\documentclass[sigconf]{acmart} 
\AtBeginDocument{%
  }

\usepackage{amsmath}
\usepackage{xurl}
\usepackage{hyperref}
\usepackage{tablefootnote}
\usepackage{algorithm}
\usepackage{algorithmicx}
\usepackage[noend]{algpseudocode}
\usepackage{mathtools}
\usepackage{caption}
\usepackage{balance}
\usepackage{bbm}
\usepackage{multirow}
\usepackage{tikz}
\usepackage[framemethod=TikZ]{mdframed}
\usepackage{xspace}
\usepackage{enumitem}
\usepackage{xcolor}
\usepackage{float}
\usepackage{tablefootnote}
\usepackage{tabularx}
\usepackage{booktabs}
\definecolor{ao}{rgb}{0.0, 0.5, 0.0}
\usepackage{listings,multicol}  
\usepackage{graphicx}
\usepackage{subcaption}
\usepackage{float} 

\usepackage{courier}
\usepackage{adjustbox}
\usepackage[most]{tcolorbox}

\newcommand{\sys}{\textsc{Hoss}\xspace}
\newcommand{\base}{\textsc{Compass}\xspace}
\newcommand{\pmm}{\textsf{Pmem}\xspace}

\newcommand{\eat}[1]{}

\usepackage{subcaption}

\newtoggle{notes}
\toggletrue{notes} 

\newcommand{\re}[1]{{\color{black} #1}}

\definecolor{codebg}{HTML}{F7F7F9}
\definecolor{cm}{RGB}{0,128,0}     
\definecolor{kw}{RGB}{0,0,160}     

\lstdefinestyle{osdi}{
  language=C++,
  basicstyle=\ttfamily\footnotesize,
  numbers=left, numberstyle=\tiny\color{gray},
  xleftmargin=2em, framexleftmargin=2em,
  frame=single, rulecolor=\color{black!15},
  backgroundcolor=\color{codebg},
  showstringspaces=false, columns=fullflexible, keepspaces=true,
  keywordstyle=\color{kw}\bfseries,
  commentstyle=\color{cm}\itshape,
  tabsize=2, upquote=true,
  morekeywords={
    View,filterAndConditions,select,groupBy,count,sort,
    filter,agg,orderBy,asc,desc,col,lit,as
  }
}
\makeatletter
\DeclareFontFamily{OMX}{MnSymbolE}{}
\DeclareSymbolFont{MnLargeSymbols}{OMX}{MnSymbolE}{m}{n}
\SetSymbolFont{MnLargeSymbols}{bold}{OMX}{MnSymbolE}{b}{n}
\DeclareFontShape{OMX}{MnSymbolE}{m}{n}{
    <-6>  MnSymbolE5
   <6-7>  MnSymbolE6
   <7-8>  MnSymbolE7
   <8-9>  MnSymbolE8
   <9-10> MnSymbolE9
  <10-12> MnSymbolE10
  <12->   MnSymbolE12
}{}
\DeclareFontShape{OMX}{MnSymbolE}{b}{n}{
    <-6>  MnSymbolE-Bold5
   <6-7>  MnSymbolE-Bold6
   <7-8>  MnSymbolE-Bold7
   <8-9>  MnSymbolE-Bold8
   <9-10> MnSymbolE-Bold9
  <10-12> MnSymbolE-Bold10
  <12->   MnSymbolE-Bold12
}{}

\let\llangle\@undefined
\let\rrangle\@undefined
\DeclareMathDelimiter{\llangle}{\mathopen}%
                     {MnLargeSymbols}{'164}{MnLargeSymbols}{'164}
\DeclareMathDelimiter{\rrangle}{\mathclose}%
                     {MnLargeSymbols}{'171}{MnLargeSymbols}{'171}
\makeatother

\newcounter{theo}[section] 
\renewcommand{\thetheo}{\arabic{section}.\arabic{theo}}

\newcounter{lekg}[section] 
\newcommand{\thelkg}{\arabic{section}.\arabic{lekg}}

\newcounter{prot}[section] 
\renewcommand{\theprot}{\arabic{section}.\arabic{prot}}

\setcopyright{acmlicensed} 
\copyrightyear{2018} 
\acmYear{2018} 
\acmDOI{XXXXXXX.XXXXXXX} 
\acmConference[CCS '26]{Make sure to enter the correct
  conference title from your rights confirmation email}{November 15-19,2026}{Netherlands}  
\acmISBN{978-1-4503-XXXX-X/2018/06}  

\begin{document}

\title{\sys: Fast Oblivious Semantic Search with Heterogeneous GPU-CPU-TEE Architecture} 


\setcopyright{none}
\settopmatter{printacmref=false}
\renewcommand\footnotetextcopyrightpermission[1]{}
\settopmatter{authorsperrow=4}
\author{Jianzhang Du}
\authornote{The first two authors contributed equally to this research.}

\author{Weijie Huang}
\authornotemark[1]

\author{Chenghong Wang}

\affiliation{%
  \institution{Indiana University}
  \country{}
}
\email{{du5,wh25,cw166}@iu.edu}

\author{Nicolas Tsagareli}
\email{ntsagareli@binghamton.edu}

\author{Yukui Luo}
\email{yluo11@binghamton.edu}

\affiliation{%
  \institution{Binghamton University}
  \country{}
}

\author{XiaoFeng Wang}

\affiliation{%
  \institution{Nanyang Technological University}
  \country{}
}
\email{xiaofeng.wang@ntu.edu.sg}

\author{Zhongshu Gu}

\affiliation{%
  \institution{IBM Research}
  \country{}
}
\email{zgu@us.ibm.com}







\renewcommand{\shortauthors}{}

\begin{abstract}
Semantic search is widely deployed in modern AI systems, but protecting both data contents and access patterns remains challenging. The current state-of-the-art system, \textsc{Compass}, achieves oblivious semantic search by building an optimized ORAM over HNSW graphs. However, even with aggressive optimizations, it still incurs large overheads. Closing this performance gap is fundamentally difficult: \textsc{Compass} has already removed most cryptographic overheads, leaving ORAM accesses as the dominant cost, which are constrained by well-known $\Omega(\log N)$ bandwidth lower bounds.

Our key insight is that traditional ORAM overhead stems from the assumption of limited private memory, whereas modern GPU TEEs provide large private memory (\pmm) that blinds internal access patterns~({\em Hunt et al., NSDI’23}). This shift opens a new design space. We therefore propose \sys, a first-of-its-kind oblivious semantic search system with a heterogeneous CPU–GPU TEE architecture that supports fast, scalable search with low cost of ownership. In \sys, the GPU TEE's large \pmm hosts the hot-path HNSW traversal, while the lower layers of the graph, if exceeds GPU capacity, are offloaded to CPU TEEs. The system invokes oblivious primitives only when accessing these lower layers. The availability of large \pmm also enables new optimization opportunities. For example, \sys features a host-access ORAM mechanism that goes beyond traditional performance constraints, and incorporates several data-dependent optimizations that are not possible in prior designs. We implement a prototype of \sys and benchmark it against \textsc{Compass}. Our result shows that \sys achieves \re{up to $99\times$ speedup} while maintaining high recall, with larger gains at scale.
\end{abstract}




\maketitle

\section{Introduction}
\label{sec:intro}

Semantic search~\cite{douze2024faiss, reimers2019sentence, malkov2018efficient, xiong2020approximate, chen2021spann, pan2023survey, bast2016semantic, mangold2007survey} maps queries and data into a shared embedding space and retrieves results by similarity. It has become a core primitive in modern AI systems, supporting recommendation engines~\cite{bobadilla2013recommender}, retrieval-augmented generation (RAG)~\cite{gao2023retrieval}, and personalized search~\cite{liu2020personalization}. It is also critical in pharmacology, chemistry, and biomedicine, where it enables similarity-based discovery over chemical compounds, molecular structures, genetic sequences, and patient records. However, existing cloud deployments are unsuitable for these high-stakes domains because of privacy, trust, and proprietary constraints. Most systems operate on plaintext, exposing both queries and embeddings to service providers. Although encrypted search techniques~\cite{wang2016secure, curtmola2006searchable, demertzis2020dynamic, demertzis2018searchable} support queries over encrypted data, they do not eliminate leakage in semantic retrieval: because search is data-dependent, observable memory-access patterns can still reveal structural information about query semantics and user intent~\cite{ding2025moecho, zhu2025compass, kellaris2016generic, blackstone2019revisiting}.

Recent work addresses this leakage through data-oblivious retrieval~\cite{zhu2025compass, engelsma2022hers, chen2020sanns, zhou2024pacmann, li2025panther, januszewicz2026ragtime}. However, most approaches either incur high overhead from heavy cryptographic primitives such as fully homomorphic encryption or garbled circuits~\cite{li2025panther, chen2020sanns, engelsma2022hers, januszewicz2026ragtime}, or remain largely theoretical~\cite{zhou2024pacmann, feng2025enabling}. \textsc{Compass}~\cite{zhu2025compass}, the state of the art, provides end-to-end oblivious semantic search with reasonable accuracy and scalability by building an optimized ORAM for HNSW indexes~\cite{malkov2018efficient}. Yet even with extensive optimizations, \textsc{Compass} still incurs seconds-level latency on million-scale datasets at roughly $0.9$ recall, whereas plaintext systems can achieve $>0.95$ recall with microsecond-level latency~\cite{qdrant_benchmarks}. Closing this gap is fundamentally difficult: \base already removes most heavy cryptographic costs, leaving ORAM accesses as the dominant bottleneck, as also observed in other oblivious data-processing systems~\cite{chamani2023graphos, zheng2024h, oblix, tinoco2023enigmap}. Given the well-known lower bounds of ORAMs~\cite{goldreich1987towards}, the remaining room for optimization is limited.

\vspace{4pt}\noindent{\bf A new opportunity.} Recent GPU-based Trusted Execution Environments (TEEs)~\cite{costan2016intel} open up a new opportunity. TEEs provide hardware-isolated execution that sandboxes sensitive code and data from the rest of the software stack. They have been widely used to build secure retrieval systems~\cite{oblix, chamani2023graphos, eskandarian19obliDB, zheng2017opaque, tinoco2023enigmap, Ahmed2025oasisdb}. Traditional CPU TEEs still rely on ORAM to achieve strong obliviousness. This is because memory access patterns can leak through shared microarchitectural states to co-located tenants~\cite{gras2018tlbleed, Kocher2018spectre}, or be inferred by adversaries observing memory traces~\cite{chuang2026tee}. However, GPU TEEs have changed this landscape. Unlike CPUs, GPUs execute kernels in a more dedicated manner. GPU TEEs default to a {\em full-device pass-through} mode~\cite{nvidia_confidential_computing_h100}, dedicating the entire GPU as an enclave to a single tenant. This minimizes cross-tenant interference. Moreover, modern GPUs integrate on-package 3D-stacked memory (e.g., 80GB of HBM~\cite{nvidia2025secureai}) tightly coupled with compute cores. There are no exposed external memory channels. As a result, even physical snooping of memory traffic becomes extremely hard~\cite{nvidia_confidential_computing_h100, volos2018graviton, gu2026blueprint}.


In classical ORAM studies, dedicated (non-shared), on-package memory is assumed to serve as private memory (\pmm), whose access patterns are not visible to the adversary~\cite{goldreich1996software}. In CPU settings, however, only registers satisfy this requirement, so that prior ORAM studies to assume merely $O(1)$-sized \pmm~\cite{goldreich1996software, asharov2020optorama}. GPU TEE’s isolated HBM meets the same requirement—it is on-package and dedicated to a single trusted workload—making it reasonable to treat as a much larger \pmm. Both industry~\cite{nvidia2025secureai} and academia~\cite{volos2018graviton} have recognized this capability. Recently, it has led to a new class of oblivious systems~\cite{hunt2020telekine,guo2025bolt} and opens a new design point that reduces the traditional ``ORAM tax.''


\vspace{4pt}\noindent{\bf The problem of this work.} Our key insight is that large \pmm in GPU TEEs reshapes the design space of oblivious semantic search. We therefore explore how to leverage this capability to rethink existing designs and ask the following central question:
    {\em Can \pmm in GPU TEEs enable new system designs that match the security of existing oblivious semantic search (e.g., \base), but with significantly lower overhead?}
More concretely, we follow a setting as \base, focusing on HNSW-based semantic search, a top-performing and widely adopted method with strong accuracy in practice\footnote{That said, our approach is not limited to HNSW and can be extended to other hierarchical graph index methods (see \S~\ref{sec:discussion}).}~\cite{zhu2025compass, malkov2018efficient}. HNSW organizes embeddings into a multi-layer graph, where nodes at each layer is an embedding instance and form a subset of the layer below, with the bottom layer containing all embeddings. Search follows a greedy, layer-by-layer descent, using upper layers for coarse routing and lower layers for fine-grained nearest neighbor refinement. Building on this basis, our goal is to leverage \pmm in modern GPU TEEs to build a secure outsourced retrieval system that, upon receiving a query from the data owner, efficiently traverses the HNSW graph to locate the top similar embeddings. The system must protect both the contents of the data and the memory access patterns during HNSW graph traversal. We also seek the system to maintain high recall rate (accuracy goal, {\bf G-1}), low search latency (performance goal, {\bf G-2}), and scalable to real-world datasets with Gigabytes or even Terabytes embeddings (scalability goal, {\bf G-3}). In addition, we aim to keep the total cost of ownership (TCO) low when leveraging GPUs (cost goal, {\bf G-4}).

\vspace{4pt}\noindent{\bf Unique challenges.} Despite the availability of large \pmm, achieving the above goals together require non-trivial designs. For example, meeting G-3 with a single GPU TEE is challenging, as real-world semantic datasets can easily exceed the HBM capacity of one GPU (e.g., 80 GB on an H100). A straightforward approach is to scale out to multi-GPU TEE deployments~\cite{nvidia2025secureai}. However, data retrieval workloads, including semantic search, are inherently memory-bound rather than compute-bound~\cite{li2025scaling}, so scaling out often leads to underutilized compute resources. Moreover, GPU TEE compute resources cannot be flexibly shared across workloads, further inflating TCO and making it difficult to satisfy G-4. 

\subsection{Our Contributions}
To address these challenges and meet our design goals, we propose \sys, the first {\bf H}eterogeneous-TEE-based system for {\bf O}blivious {\bf S}emantic {\bf S}earch. Figure~\ref{fig:overview} provides an overview of the system. 
\begin{figure}[tb]
    \centering
    \includegraphics[width=0.8\linewidth]{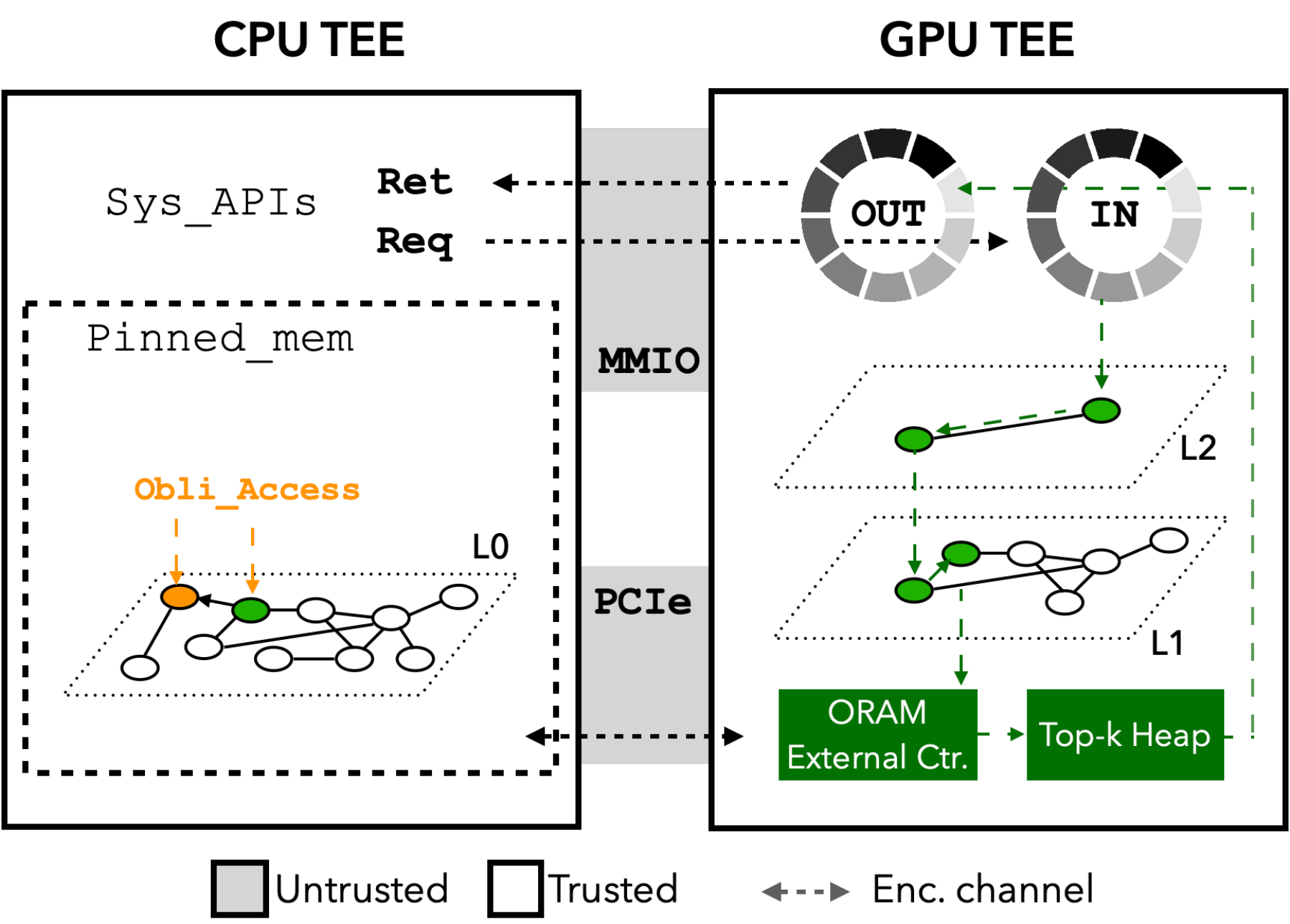}
    \caption{\sys overview. }
    \label{fig:overview}
\end{figure}
At its core, \sys adopts a heterogeneous design and avoids relying on multi-GPU TEEs. It offloads the HNSW graph, especially the lower layers, to the host CPU TEE, which offers a much larger secure memory space (e.g., up to terabytes). We partition the graph by layers, as the hierarchical structure of HNSW provides clean boundaries for both data and search. This layer-wise offloading simplifies both the system and the search algorithm. The upper-layer search runs entirely within the GPU TEE's \pmm, which allows fast data-dependent operations on the hot path. When the search descends to the lower layers, the system accesses host memory through a special ORAM controller inside the GPU TEE. With this design, all outsourced data remains protected within TEEs: data-dependent accesses are confined to the GPU TEE’s \pmm, while accesses to host memory are made oblivious. As such, both data contents and access patterns are protected, and meets our security goals. The heterogeneous design also lets the system scale to terabyte-scale datasets (G-3) with only a single GPU (G-4). In addition, the large \pmm provides us opportunities to introduce new ORAM designs that avoid traditional bottlenecks. This plays a key role in achieving efficient query capability (G-2) without lowering recall (G-1). We summarize our technical contributions below:
\begin{itemize}[leftmargin=*, itemsep=2pt]
    \item \textbf{\em Architecture.} We propose the first heterogeneous TEE-based oblivious semantic search system that delivers high efficiency, accuracy, and scalability while maintaining low TCO.
\item \textbf{\em Storage layout.} We propose a sorted, linear table representation for HNSW graphs that largely reduces heavy pointer-based adjacency structures, and enables efficient GPU execution and ORAM mapping. We map this logical design to concrete storage layouts across CPU and GPU TEE environments.
    
  \item \textbf{\em Execution model.} We propose a self-hosted execution model that repurposes CPU bypassing for security. Instead of using the CPU to orchestrate execution, the GPU maintains all query state and control flow, which reduces host interference.

    \item \textbf{\em New ORAM.} The large \pmm in GPU TEEs (e.g., $O(\alpha \cdot N)$, where $\alpha>0$ and $N$ denotes the total data size) opens new opportunities for ORAM design. We introduce an ORAM mechanism that supports access to host memory with only \smash{$O\left(\frac{\log N}{\log\log N}\right)$} bandwidth overhead. In contrast, under the traditional ORAM model with $O(1)$ \pmm, the lower bound is \smash{$\Omega\left(\log N\right)$}. We further optimize the design with a secure coalescing mechanism and provide theoretical analysis to validate its security and bandwidth gains.
  
    \item \textbf{\em Prototype and evaluation.} We implement \sys in CUDA C++ with 8.1K lines of code and conduct a comprehensive evaluation against the SOTA systems. Our results show that \sys achieves \re{up to $99\times$ speedup} and delivers larger gains at scale, with improvements increasing for larger datasets.
\end{itemize}


\section{Background}
\label{sec:background}

\subsection{Similarity Search}\label{sec:cmp-ss}



Similarity search is a prominent method for retrieving data from large databases. It is extensively used in search~\cite{halavais2017search} and recommendation systems~\cite{ko2022survey}. Recently, it has become increasingly prevalent in critical fields such as biomedical research, pharmaceutical discovery~\cite{lo2018machine, cheng2011identifying}, and chemical informatics~\cite{lopez2024molecular,willett1998chemical}. In similarity search, entities such as images, documents, or molecules are mapped into high-dimensional vectors - embeddings - and retrieval is performed by finding nearest neighbors under a distance metric. Semantic search is a prominent instance of similarity search that captures semantic meaning learned from data.

Finding exact nearest neighbors becomes computationally intractable at large scale and high dimensionality, making exhaustive search impractical in real-world deployments. This is addressed by adopting approximate nearest neighbor (ANN) search techniques, which sacrifice a small amount of accuracy for substantially greater efficiency. ANN methods use specialized indexing structures to accelerate queries without scanning all data points. Indexing strategies developed over the years include tree-based methods~\cite{bentley1975multidimensional,beygelzimer2006cover}, locality-sensitive hashing~\cite{sharma2018improving}, vector quantization~\cite{wu2024structured}, and graph-based approaches~\cite{malkov2020efficient}. Among these, graph-based methods have received substantial attention due to their strong accuracy, high performance, and applicability across diverse workloads.
\vspace{5pt}\noindent{\bf HNSW~\cite{malkov2018efficient}} 
is one of the most popular and effective graph-based ANN indexing 
methods~\cite{zhu2025compass, aumueller2017ann}. It organizes data points in a multi-level graph, where higher layers 
give a coarse overview and lower layers capture finer neighborhoods. Search starts at an entry node on the 
top layer and moves down level by level. At upper levels, a greedy walk follows a single best candidate, 
steering the query toward relevant regions. At the bottom layer, the search expands into a dynamic candidate 
list of size \emph{ef}, exploring multiple paths instead of just one. \emph{ef} is a key parameter - it 
controls how broadly the algorithm explores at the base layer, directly trading off accuracy against latency. 
A bigger \emph{ef} means more candidates are checked, giving higher recall but costing more compute. 

\subsection{TEEs and ORAMs} 
\noindent{\bf TEEs and access pattern leakages.} TEEs are hardware-protected regions of a processor 
that ensure code and data loaded inside them remain isolated from the rest of the 
system, including privileged software such as the operating system or hypervisor. 
Examples include Intel TDX and AMD SEV. While TEEs guarantee confidentiality and 
integrity of enclave memory, a notorious security pitfall is the \emph{access pattern leakage}:  the sequence of memory operations (e.g., \texttt{Load} and \texttt{Store}) performed by enclave programs can still be inferred by a malicious observer. Access pattern leakage arises in two primary ways. First, adversaries may exploit 
\emph{shared microarchitectural resources} such as caches, branch predictors, or TLBs 
to infer sensitive program behavior through timing differences~\cite{Kocher2018spectre, branch_predictor, yarom2014flush+, gras2018tlbleed}. 
Second, even if on-chip state is protected~\cite{xu2019hermetic}, \emph{external memory channels} remain 
visible: DRAM buses, PCIe links, DMA buffers, and even board-level copper traces 
can all be snooped or tampered with 
\cite{pessl2016drama, hu2020deepsniffer-2, gross2019breaking, zuo2020sealing}. 

ANN searches exhibit highly data-dependent memory traces, which makes them especially susceptible to access-pattern leakages~\cite{jia2025found,zhu2025compass}. Hence, relying solely on CPU TEEs to secure these algorithms is insufficient.

\vspace{5pt}\noindent{\bf ORAMs.} 
ORAM~\cite{goldreich1987towards, path-oram,ren2015constants, resizable-tree-based-oram, bindschaedler2015practicing, asharov2023futorama, asharov2020optorama, dittmer2020oblivious, zheng2024h} is a cryptographic primitive designed to hide memory access patterns in the RAM model. In its simplest form, memory is modeled 
as a sequence of address-value pairs $(\mathsf{addr}_i, v_i)$ with consecutive integer addresses. An ORAM ensures that every read or write operation results in a pseudorandom sequence 
of physical accesses, preventing an adversary from learning which exact memory line was targeted. Oblivious Maps (OMAPs)~\cite{zheng2024h, oblix, chamani2023graphos, guo2025bolt} generalize this notion to key-value stores with arbitrary keys 
(e.g., strings, sparse indexes) but comes with additional costs, typically in I/O bandwidths.

Traditional ORAMs follow a client-server model \cite{path-oram}, where a trusted 
client is responsible for memory obfuscation tasks such as shuffling and remapping memory blocks, while the server merely stores encrypted in-memory data. This design, however, imposes significant client-side overhead and high 
communication costs. To our knowledge, the only oblivious ANN system, Compass~\cite{zhu2025compass} is built upon this costly setup. Modern ORAMs reduce this burden by using CPU TEEs as the 
trusted controller and manages its private memory in ORAM subroutines. Unfortunately, even the 
obfuscation logic itself can 
leak critical access patterns. To address this, Mishra et al.\ introduced the notion of 
\emph{doubly-obliviousness (DO)} \cite{oblix}, which requires 
concealing the access patterns of both the data and the ORAM control program. 
In this work, we adopt the DO outsourcing paradigm as our default model and 
refer to it simply as \emph{obliviousness}.



\vspace{5pt}\noindent{\bf GPUs and GPU TEEs.} Modern GPUs expose massive parallelism through thousands of hardware threads organized into streaming multiprocessors (SMs). Threads execute in lockstep groups called warps (32 lanes on NVIDIA GPUs), where all lanes follow the same instruction stream~\cite{kayvon2004understanding}. Warps are further grouped into cooperative thread arrays (CTAs) that share fast on-chip scratchpad memory, and persistent kernels can keep CTAs resident to sustain high throughput. This execution model strongly favors structured, SIMD-style computation~\cite{kayvon2004understanding}. When threads within a warp diverge due to branches, execution becomes serialized, reducing efficiency~\cite{cuda_programming_guide, volkov2010better}. Similarly, irregular memory accesses degrade bandwidth utilization. As a result, GPUs perform best with uniform control flow and coalesced memory accesses, and tend to perform poorly with branch-heavy or irregular workloads~\cite{cuda_programming_guide}.

Recent work has extended TEEs beyond CPUs to accelerators, most notably GPUs~\cite{volos2018graviton, nvidia_confidential_computing_h100, hunt2020telekine, mai2023honeycomb}. While the original motivation was largely performance-driven, researchers have observed that GPU TEEs also provide strong resistance to access-pattern leakage due to their architectural design~\cite{volos2018graviton, hunt2020telekine, guo2025bolt}. First, GPUs are typically dedicated to a single tenant during secure execution~\cite{nvidia_confidential_computing_h100}, avoiding the microarchitectural resource sharing that enables many side channels on CPUs. Second, modern GPU TEEs incorporate large on-package HBM. Unlike CPU TEEs, where traffic to external DRAM remains observable, GPU HBM is integrated directly on the package and cannot be externally snooped. As a result, memory accesses within GPU HBM are effectively hidden even from a physical adversary~\cite{volos2018graviton, hunt2020telekine, guo2025bolt, gu2026blueprint}. Guo et al.~\cite{guo2025bolt} recently introduced a new ORAM design that exploits secure HBM, breaking through the long-standing bandwidth lower bound of traditional ORAMs.

\section{Threat Models \& Security Goals}
\label{sec:overview}

We formulate our system in the standard secure outsourced computing model. At a 
high level, the framework involves three logical entities: (i) a \emph{data owner}, 
who possesses a semantic dataset $D$ (in HNSW graph style) and wishes to outsource its storage and search 
to an untrusted cloud; (ii) a \emph{cloud server}, which manages computing 
resources (including TEEs) to deliver confidential semantic search services; and 
(iii)a \emph{vetted analyst}, authenticated to access $D$, who issues a semantic query $\texttt{Search}(q, D, k)$ to retrieve the top-$k$ embeddings in $D$ that are most similar to the input query $q$. \re{This work focuses on secure search over a pre-built HNSW index (e.g., by the data owner), as efficient oblivious HNSW search is already a non-trivial challenge. Other index-management primitives, such as insertion and deletion, naturally build upon the same search primitive and are discussed in \S~\ref{sec:discussion}.}



\subsection{Threat Model}

We adopt the same threat model as prior accelerator TEE work~\cite{hunt2020telekine,guo2025bolt,volos2018graviton} and industry specifications~\cite{nvidia_confidential_computing_h100}. We trust the cloud provider’s organizational integrity but treat its software stack, administrators, and co-located tenants as untrusted. The adversary is powerful: it may compromise any software layer and obtain physical access to hardware, enabling passive observation of off-chip channels. In particular, we assume all exposed interconnects, including DMA buffers~\cite{gross2019breaking}, host DRAM~\cite{pessl2016drama,chuang2026tee}, and PCIe links~\cite{hu2020deepsniffer-2}, can be snooped. We exclude chip-level attacks such as depackaging or probing silicon interposers~\cite{nvidia_confidential_computing_h100}, as well as active physical side channels (e.g., power analysis~\cite{xiang2020open,li2022power} and electromagnetic emanations~\cite{gongye2024side}) that require invasive access, such as remove heat sinks and unseal chip packages, and are impractical in data center settings. We also assume the GPU TEE operates in full passthrough mode and is dedicated to a single program (e.g., \sys), as is typical in current deployments~\cite{nvidia_confidential_computing_h100}; therefore, co-location attacks that rely on sharing a GPU with a victim~\cite{nayak2021mis,zhang2023t,nazaraliyev2025not,zhang2024invalidate,almusaddar2025shadowscope} are out of scope. We do not consider availability attacks~\cite{luo2020stealthy}, covert channels~\cite{miao2024veiled}, as they're not targeting the confidential computing guarantees. We also do not consider attacks via malicious user inputs~\cite{guo2024gpu}, as these fall outside the threat model of secure outsourced computing, where users are assumed to be data owners themselves or their authorized users. Lastly, 
we assume the attacker cannot break symmetric or public-key cryptosystems.

\subsection{Security Goals}
\label{sec:security-definition}
We follow the standard obliviousness definitions used in many ORAM papers~\cite{path-oram, ren2015constants, asharov2020optorama}. We slightly adapt it in the context of semantic searches. Intuitively, given an HNSW graph dataset $D$ and a query $q$, we define $\textsf{path}(q, D) \gets \bigcup_{i=0}^{k} \mathbf{r}_i$ as the logical traversal path for answering $q$, where $\mathbf{r}_i$ denotes the sequence of logical node accesses at layer $i$ of the HNSW graph. Our security notion requires that for any two queries $q$ and $q'$ with the same traversal lengths, i.e., $\forall i \in \{0, \dots, k\},\ |\mathbf{r}_i| = |\mathbf{r}'_i|$, no probabilistic polynomial-time (p.p.t.) adversary in our threat model can distinguish between their executions by observing the leakage transcripts. This is analogous to classical ORAM definitions, where security is defined over logical access sequences of same length~\cite{path-oram}. Specifically,

\begin{definition}[Obliviousness in HNSW search]\label{def:security}
For any p.p.t. adversary $\mathcal{A}$, and for any two queries $q$ and $q'$ over the same HNSW data $D$, such that $\forall i \in \{0, \dots, k\},\ |\mathbf{r}_i| = |\mathbf{r}'_i|$, we have
\[
\left| 
\Pr[\mathcal{A}(\textsf{View}(q, D)) = 1] - 
\Pr[\mathcal{A}(\textsf{View}(q', D)) = 1]
\right| \le \textsf{negl}(\lambda),
\]
where $\textsf{View}(\cdot)$ is the leakage transcript, and $\textsf{negl}(\lambda)$ denotes a negligible probability under the security parameter $\lambda$.  
\end{definition}




\subsection{Infrastructure Assumptions}
We make several infrastructure-level assumptions. First, GPU TEE memory is limited: the footprint of ANN indexes--including both structural metadata and embeddings--can exceed the capacity of a single GPU's HBM (e.g., 80 GB on an H100 TEE)~\cite{widmoser2025shine, munyampirwa2024down}. We do not assume the ability to stitch multiple GPUs into a larger trusted pool, as current vendors provide no such support and it would require new hardware features~\cite{nvidia_confidential_computing_h100}. While multi-GPU TEEs are an interesting direction, they are orthogonal to our focus. We assume that host-side private memory (e.g., CPU TEEs) is comparatively larger and can hold the remainder of the index. Under this model we study a purely in-memory setting, though our design naturally extends to persistent storage at even larger scales. Concretely, we parameterize capacity as $\texttt{size\_of}(M_{\mathsf{hbm}}) = \alpha \cdot \texttt{size\_of}(M_{\mathsf{host}})$, where $M_{\mathsf{hbm}}$ and $M_{\mathsf{host}}$ denote GPU and CPU TEE memory, respectively, and $\alpha \in (0,1)$. Finally, we assume standard TEE protections, including memory encryption, integrity verification, replay defense, and encrypted inter-TEE communication~\cite{volos2018graviton,nvidia_confidential_computing_h100}. We omit these baseline mechanisms and focus on the core contributions of \sys.

\section{\sys System Design}
\label{sec:indexes}
In this section, we present the technical design of \sys. We first describe the storage layout for HNSW graphs in \sys (\S~\ref{sec:storage}). We then take a top-down view, starting from the top-layer searches in GPU TEEs (\S~\ref{sec:top}) and proceeding to the new ORAM design that support fast bottom-layer accesses on CPU TEE (\S~\ref{subsec:oram}).


\subsection{Storage Layout}\label{sec:storage}
Perhaps the most fundamental question we must address is \emph{how to store the HNSW graph}. This is non-trivial. Common graph representations, such as adjacency matrices and adjacency lists, do not fit our setting well. Adjacency matrices scatter data across memory, leading to fragmented layouts and poor utilization of GPU HBM. Adjacency lists rely on pointer chasing and irregular access patterns, which are inefficient on GPUs. Moreover, neither representation maps well to ORAM semantics, which assume a compact address space and support only fixed-size address–value accesses (i.e., accessing data by logical address over contiguous blocks). A direct mapping would create sparse address spaces and require expensive OMAPs~\cite{guo2025bolt, zheng2024h, tinoco2023enigmap} instead of simpler ORAM schemes. These then motivate us to design a new representation that better matches both GPU execution and ORAM access models.

\eat{
\begin{itemize}[leftmargin=*, itemsep=2pt]
    \item {\em Minimize I/O overhead.}  
    Communication between CPU and GPU TEEs is costly: each transfer requires re-encryption, decryption, and integrity checks. Current GPU TEEs cannot directly DMA host private memory, forcing all data to pass through a bounce buffer~\cite{nvidia_confidential_computing_h100}, which adds an extra memory copy. Our first goal is to design a data partitioning plan that keeps I/O small so that search can run efficiently.  
    \item {\em Maximize HBM efficiency.}  
    GPU HBM is scarce but extremely valuable for data-dependent computation. Standard graph layouts (e.g., adjacency lists) are sparse and often create fragmentation, wasting HBM capacity. Our second goal, thus, is to design data layouts that reduce fragmentation and fully exploit HBM while aligning with the GPU’s architecture.  
    \item {\em ORAM-friendly representation.}  
    The graph representation must also be compatible with ORAM’s address–value access model. A na\"ive key–value encoding could work but still problematic: intermediate-layer nodes span discontinuous regions of the graph, creating sparse address spaces (Figure~\ref{fig:flow}.1) that require expensive OMAPs rather than simpler ORAM schemes. Therefore, our third goal is to develop a graph representation that is compact and ORAM-compatible.  
\end{itemize}
}
\begin{figure}[tb]
    \centering
    \includegraphics[width=0.8\linewidth]{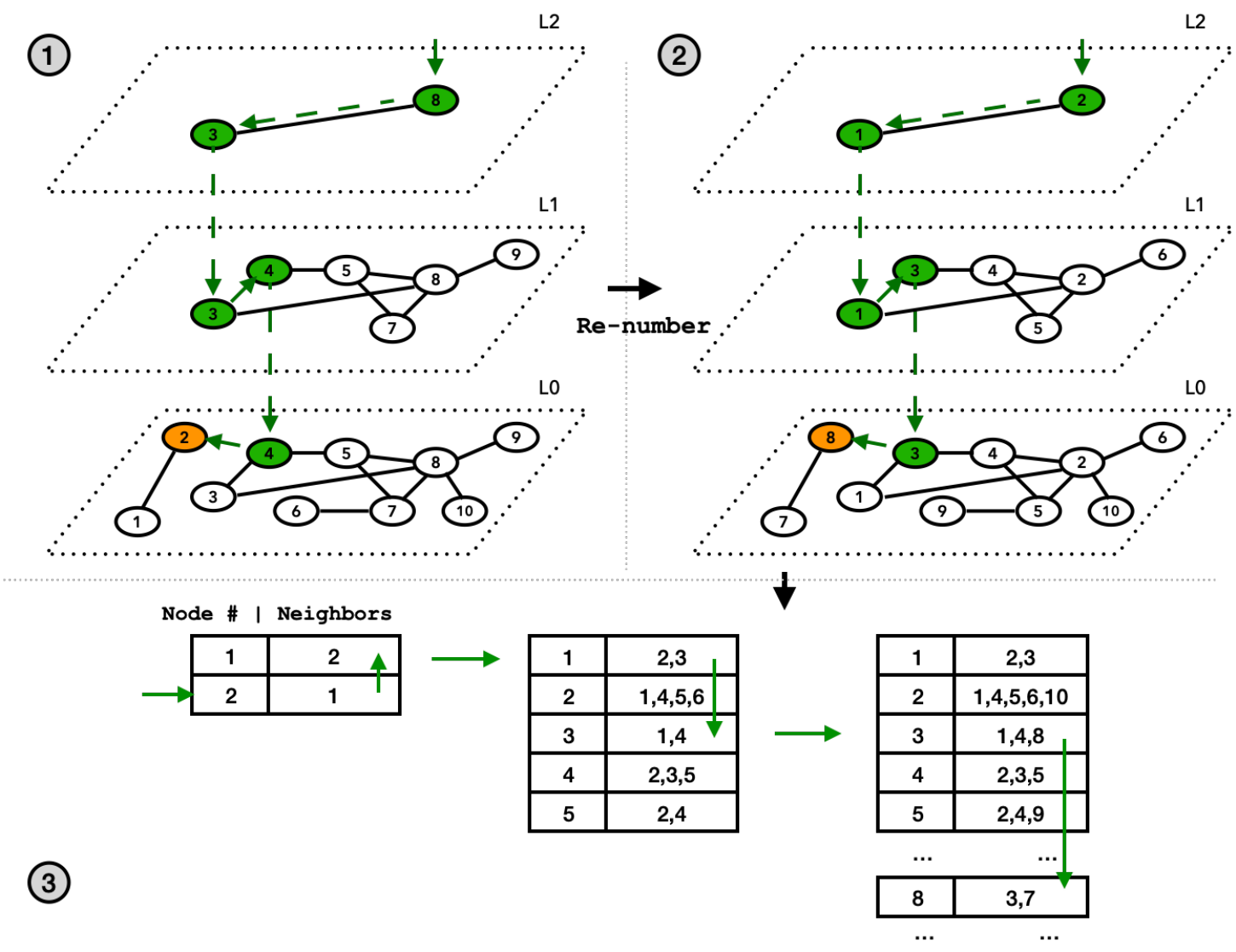}
    \caption{Sorted, decomposed linear table representation}
    \label{fig:flow}
\end{figure}

\vspace{5pt}\noindent{\bf Sorted linear tables for HNSW.} To be efficient for both GPU execution and ORAM semantics, we introduce a novel yet surprisingly simple graph representation. The key idea is to flatten each layer of the hierarchy into a compact linear table, sorted by node ID. Figure~\ref{fig:flow} shows the layout.

Every node is assigned a temporary ID that stays the same across all layers the node appears in. These IDs are allocated top-down: nodes first appearing at the highest layer get the smallest IDs, then the next layer down, and so on until layer 0. Since each layer is a superset of the layer above it, a node that appears at layer $\ell$ also appears at layers $\ell-1, ..., 1, 0$, always under the same ID. Moreover, in HNSW graphs, each node has a bounded number of neighbors bounded by parameter $M$. This means that we can allocate a fixed-size slot for each node to store its neighbor list. Together, these features lead to a simple and uniform addressing scheme. Given a node ID, we can locate its entry in any layer by direct indexing, i.e., by computing $\texttt{base\_addr} + \texttt{ID} \times \texttt{size\_of(node)}$. As a result, both within-layer access and cross-layer access take $O(1)$ time. This addressing model matches ORAM semantics well and also removes pointer chasing, which is especially important for efficient GPU execution.

This representation is simple, and so too is the process of generating it, as illustrated in Figure~\ref{fig:flow}. We start with a hierarchical graph whose nodes are identified by arbitrary IDs (e.g., document IDs) (Figure~\ref{fig:flow}.1). The key step is to assign each node a new temporary ID that is consistent across all layers where the node appears. These temporary IDs are allocated in ascending order, proceeding layer by layer from the topmost to the bottom (Figure~\ref{fig:flow}.2). Once assigned, constructing the linear tables becomes straightforward: each layer is treated as an independent graph, and we record its ID-to-neighbor mappings, sorted by the temporary ID.

\vspace{5pt}\noindent{\bf Physical storage.} With the sorted linear-table abstraction in place, we now describe how data is physically laid out across CPU memory and GPU TEE memory. 

As noted earlier, we partition storage along layer boundaries. For upper layers that reside in GPU HBM, we use a decomposed layout: the graph structure is stored in one contiguous array, while embeddings are stored in a separate contiguous embedding slab. This avoids duplicating large embedding vectors across layers and minimizes fragmentation. Moreover, as long as embeddings are laid out in node-ID order, we can locate an embedding by direct indexing, using the node ID. In this way, we do not need to store pointers from graph entries to embeddings. As a result, the graph structure only stores compact neighbor IDs, rather than full pointers or embedding references. This reduces metadata overhead, makes the GPU layout more compact, and frees more HBM capacity that potentially to store more upper layers.
 
For the host-resident layers, however, we use a different physical layout. There, each entry co-locates the neighbor list and the embedding in one self-contained record. The reason is simple: on the host side, the bottleneck is not storage capacity but I/O. If structure and embedding were stored separately, each node expansion would require multiple oblivious accesses. By packing them together, one ORAM access retrieves everything needed to process a node. This may duplicate some embedding data across layers, but it significantly reduces ORAM traffic, which is the more important optimization in this setting.

For the remaining host-resident layers, table entries embed both structure and embeddings directly. While this may duplicate embeddings across multiple layers when $k > 1$, we say that the host bottleneck is not memory capacity but I/O overhead when transferring memory blocks to the GPU. Co-locating structure and embeddings eliminates the need for separate ORAM lookups, and  halves I/O invocations. Moreover, this layout preserves the logical shape of linear tables, as such we can design efficient ORAM primitives ({\bf G-3}, \S~\ref{subsec:oram}) for retrieval, rather than resorting to heavyweight OMAPs.



\subsection{Execution model}\label{sec:top}
Modern GPU programs rely on a sequence of short kernels launched by the CPU~\cite{cuda_programming_guide, volkov2010better}. This model is convenient but introduces privacy risks in our setting. Frequent host interactions expose fine-grained, data-dependent execution states to the CPU~\cite{Kocher2018spectre, lipp2018meltdown, canella2019systematic}, which makes them vulnerable to side-channel attacks. In addition, CPU–GPU traffic patterns, which are observable on interconnects (e.g., PCIe), can leak strong signals about data-dependent secrets~\cite{hunt2020telekine}.

Our response is new a form of \emph{CPU bypassing}. In HPC systems, CPU bypass is typically used to reduce software overhead and move the fast path closer to the device~\cite{akram2022sok, li2019evaluating}. Here, we repurpose this idea for security. Rather than treating the GPU as a stateless accelerator, we let it self-host the search. Conceptually, a long-lived kernel maintains the full query state on device, advances the traversal across the HNSW hierarchy, and initiates and manages ORAM calls entirely from the device side for host access. It materializes results to host memory only at the end. In other words, we remove the CPU from the step-by-step control path not to reduce overhead, but to limit host interference and potential leakage.

To realize this execution model, we design a persistent search kernel with staged execution (Figure~\ref{fig:design} shows an overview). All storage is provisioned before the kernel starts: On the host side, the final-layer(s) data is placed in pinned memory and organized using an ORAM-friendly layout (\S~\ref{subsec:oram}). On the GPU side, HBM is pre-allocated for the upper-layers index and runtime structures, including input and output buffers, auxiliary data structures (e.g., visited bitmaps during graph traverses, etc.), and scratch pad memory space. The host continuously posts queries to a device-visible input buffer, and a fixed set of resident CTAs occupies the GPU and repeatedly pulls queries. The kernel then proceeds in 4 stages. 
\begin{figure}[tb]
    \centering
    \includegraphics[width=0.9\linewidth]{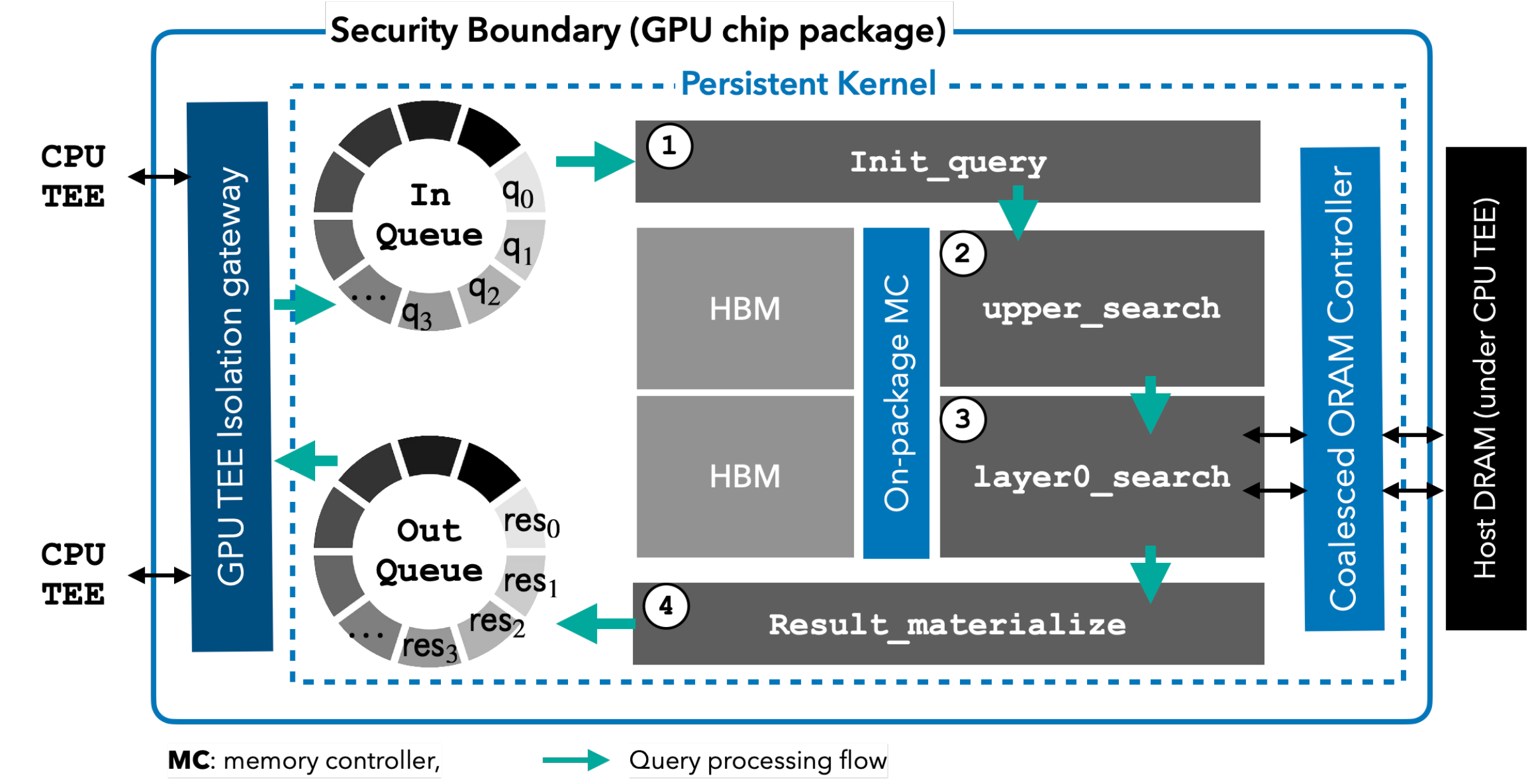}
    \caption{The \sys architecture and execution model}
    \label{fig:design}
\end{figure}

\vspace{4pt}\noindent{\bf \texttt{init\_query}.} A warp pulls a query from the input queue, decodes it, and initializes a compact per-query context on device and then execution directly transitions to the next stage.

\vspace{4pt}\noindent{\bf \texttt{upper\_search}.} The same warp first executes the top-layer HNSW search over the GPU-resident hierarchy. To accelerate the search, we also add intra-query parallelisms, where warp lanes cooperatively evaluate neighbors of the current node, process strided subsets of candidates, and compute distances in parallel.

\vspace{4pt}\noindent{\bf \texttt{layer0\_search}.} Once the traversal reaches the host-resident final layer(s), the algorithm switches from direct HBM access to device-side ORAM calls. The key point is that the kernel does not return control to the CPU. Instead, it issues ORAM-backed fetches through a GPU-resident controller (\S~\ref{subsec:oram}) that manages all ORAM logic and metadata on device. This stage preserves the same search semantics as standard HNSW search, but maps all node accesses to ORAM subroutines. One practical change is in how we realize the beam-search frontier on the GPU. The frontier is the set of active candidate nodes that the search maintains and expands at each step. In classical HNSW, it is implemented using dynamic priority queues that repeatedly extract the next candidate while updating the current best set~\cite{malkov2018efficient}. That organization works well on CPUs, but maps poorly to warp execution because it leads to irregular memory updates and branch-heavy control flow. We thus maintain the working set as bounded candidate and result arrays managed cooperatively by the warp. At each expansion, lanes scan strided subsets of the arrays, use warp-level reductions to identify the next candidate to expand, and then update the arrays in place. 


\vspace{4pt}\noindent{\bf \texttt{result\_materialize}.} After the final-layer search converges, the kernel performs final top-$k$ selection on device and writes back the final identifiers, embeddings and corresponding distances.

\subsection{GPU-Assisted Coalesced ORAM}
\label{subsec:oram}
We now introduce our novel ORAM design for efficiently accessing offloaded HNSW layer(s) in host (CPU TEE) memory.

\vspace{5pt}\noindent{\bf A first attempt--direct following of BOLT~\cite{guo2025bolt} design.}
The most straightforward approach is to adopt a classical CPU-based (doubly-oblivious) ORAM~\cite{zheng2024h,oblix,chamani2023graphos}. In this design, the GPU generates logical memory requests (e.g., which node to retrieve) and submits them to the CPU through an encrypted channel. The CPU runs a full ORAM runtime, executes the retrieval subroutine on host memory, and returns the requested record (plus dummies) to the GPU. While functionally correct, this approach suffers from the well-known $\Omega(\log N)$ I/O blowup of classical ORAM~\cite{goldreich1987towards}. In our setting, where cross-TEE I/O is already the bottleneck, such overhead is immediately prohibitive.

A natural alternative is to follow BOLT, which pushes most of the ORAM logic into the accelerator itself. In BOLT’s model, the accelerator computes logical addresses and issues them directly, leaving the host side to a lightweight role of translating addresses and performing raw memory accesses. This eliminates much of the CPU overhead. However, directly porting BOLT to our setting still performs poorly. BOLT is designed as a general OMAP rather than an ORAM, introducing machinery unnecessary for our goal. Moreover, its design relies on randomized power-of-two choices, which translate into condition-heavy execution that is inefficient on GPUs. Finally, treating HNSW simply as a generic retrieval workload fails to exploit properties unique to graph search.


\vspace{5pt}\noindent{\bf Lessons learned and our ideas.} These failed attempts suggest two important design principles. First, oblivious ANN search should exploit properties unique to HNSW rather than treating it as a generic retrieval workload. Second, oblivious algorithms should be designed around GPU execution instead of inheriting condition-heavy control logic from CPU- or hardware-oriented designs. Guided by these principles, we redesign oblivious HNSW from the ground up.

Our first insight is that HNSW enables a form of oblivious coalescing that is impossible in general retrieval workloads. \re{Na\"ively coalescing repeated accesses breaks obliviousness~\cite{path-oram}, since deterministically merging requests to the same memory location produces a distinguishable access transcript. The key observation is that the neighbors expanded from a node in HNSW are inherently distinct. Hence, when coalescing is restricted within a neighbor expansion, every access in the batch is already unique, so coalescing does not alter the access distribution (\S~\ref{sec:analysis}). Any reduction in memory accesses is indistinguishable from random collisions under independent sampling.} This HNSW-specific optimization naturally aligns with the GPU's SIMD execution model, substantially reducing cross-TEE I/O. Moreover, such coalescing fundamentally relies on graph-search semantics and therefore cannot be safely performed by the CPU, but can be exploited inside the GPU's \pmm.


Our second insight is that BOLT’s load balancing can be greatly simplified. Instead of using randomized power-of-two choices, we use uniform random remapping. This increases the asymptotic bandwidth cost from $O(\log\log N)$ to $O(\tfrac{\log N}{\log\log N})$ (see \S~\ref{sec:analysis}), where $N$ is the total data size. However, for practical data sizes (e.g., $N \leq 2^{32}$), this gap is small. In return, the hot path becomes branch-free and SIMD-friendly, which leads to larger performance gains in practice.

\vspace{5pt}\noindent{\bf Coalesced ORAM (CORAM).} We now present our CORAM design based on the above insights. For clarity, we focus exclusively on the data access path and describe CORAM without tying it to a specific HNSW layer. In general, the CORAM can be used to access any layers. To distinguish between TEE spaces, we use \textcolor{blue}{blue} text to denote structure stored in host memory (CPU TEE).

Initially, we partition host storage into $M$ fixed-sized blocks \textcolor{blue}{$V_{\mathsf{host}}[1..K]$}. Each block stores a set of records $(v, \textsf{nbrs}(v), \textsf{emb}(v))$, padded with dummy entries to maintain a uniform layout. The location of each record is tracked by a dense position map $\mathsf{MAP}[1..N]$, which resides entirely inside the GPU TEE. The GPU also maintains per-block stash queues $S[1..K]$ to temporarily hold records that have been remapped but not yet written back to host storage. We also maintain an entry cache $\mathsf{EC}$ on the GPU. For each node in the last GPU-resident layer (layer $k$), $\mathsf{EC}$ stores its neighbors in the next layer ($k{+}1$), which is offloaded to the host. This cache stores only neighbor IDs and allows the GPU to immediately start expansion once the search reaches the boundary layer.

At a high level, when expanding neighbors of a node, the GPU first looks up $\mathsf{MAP}$ to identify the blocks that contain the required records, and then reads these blocks from host memory in a coalesced batch. Once fetched, the kernel locates the target records on the GPU. Since some records may have been accessed earlier and remapped but not yet written back, they may not appear in the fetched blocks and must instead be retrieved from the stash. After processing, each accessed record is remapped to a new block. The system then writes back the blocks fetched in this round. At the same time, any records in the stash that are mapped to these blocks are merged and written back together. This workflow is summarized in Algorithm~\ref{alg:final-layer-stash-only}.



\begin{algorithm}[tb]
\caption{Coalesced ORAM for Neighbor Expansion}
\begin{algorithmic}[1]
\Statex {\bf GPU TEE:} position map $\mathsf{MAP}[1..N]$; per-block stash queues $S[1..M]$; entry$\!\to$final cache $\mathsf{EC}$; node to expand $u$.
\Statex {\bf CPU TEE:} \textcolor{blue}{$V_{\mathsf{host}}$}$[1..K]$ (page-aligned).
\Statex
\State $\mathcal{V} \gets x.\texttt{GetNeighbors}()$
\State $\mathcal{B} \gets \texttt{Distinct}\!\big(\{\mathsf{MAP}[v] \mid v\in\mathcal{V}\}\big)$
\State $\{\,\mathsf{page}_b\,\}_{b\in\mathcal{B}} \gets \texttt{READ}\!\big(\{\textcolor{blue}{V_{\mathsf{host}}}[b] \mid b\in\mathcal{B}\}\big)$

\For{$v \in \mathcal{V}$} \State $rec[v] \gets \texttt{ExtractRemove}\!\big(\mathsf{page}_{\,\mathsf{MAP}[v]}\cup S_{\,\mathsf{MAP}[v]},\, v\big)$ \EndFor
\State \texttt{ComputeOnGPU}\!$\big(\{\,rec[v]\,\}_{v\in\mathcal{V}}\big)$

\For{$v \in \mathcal{V}$} \Comment{\emph{stash-only remap}}
  \State $b' \gets \texttt{UniformRandom}(1..M)$;\quad $\mathsf{MAP}[v]\gets b'$
  \State $S[b'] \gets S[b'] \cup \{rec[v]\}$
\EndFor

\For{$b\in\mathcal{B}$} \Comment{\emph{drain only for blocks read this round}}
  \State $\textsf{pg}'_b \gets \texttt{Assemble}\!\big(\mathsf{page}_b \cup S[b]\big)$;\quad $S[b]\gets \emptyset$
\EndFor
\State \texttt{WRITE}\!$\big(\{\,(\textcolor{blue}{V_{\mathsf{host}}}[b],\ \textsf{pg}'_b) \mid b\in\mathcal{B}\,\}\big)$
\end{algorithmic}
\label{alg:final-layer-stash-only}
\end{algorithm}

\vspace{5pt}\noindent{\it Stage.1 Coalesce read (line 1:3).} To expand the neighbors of a node $x$, the GPU first obtains its neighbor set $\mathcal{V}$. This step is straightforward. If $x$ is an entry node, its neighbors are directly available in $\mathsf{EC}$. Otherwise, $x$ must have been fetched earlier, and its full neighbor list is already read into GPU. Each $v \in \mathcal{V}$ is then mapped through $\mathsf{MAP}[v]$ to its current host block, and a deduplication step produces the distinct block set $\mathcal{B}$. We store $\mathsf{MAP}$ as a dense array indexed by node IDs, which allows the GPU to efficiently process many lookups in parallel. The deduplication is also implemented using parallel primitives, which fits well with GPU execution. The result is a compact set of block IDs, which is submitted as a batch of read requests to the CPU TEE. Each request corresponds to one page-aligned \texttt{READ}, and the CPU returns all requested pages over the bounce buffer~\cite{nvidia_confidential_computing_h100}.

\vspace{5pt}\noindent{\it Stage.2 GPU local computation (line 4:7).} Once the pages are resident, the GPU extracts the requested records by invoking \texttt{ExtractRemove} on their source blocks. The search also checks the stash $S$, since some nodes may have been accessed earlier but not yet written back to host storage. Importantly, even if a node is served from the stash, the corresponding page request must still be issued in Stage 1; otherwise, an attacker could infer whether a page has been evicted. The \texttt{ExtractRemove} operation returns the record and clears its slot, so the block reflects a consistent logical state. The extracted records are then processed by the local expansion routine (\texttt{ComputeOnGPU}), which computes distances, filters candidates, and selects the next frontier entirely within GPU HBM.

\vspace{5pt}\noindent{\it Stage.3 Batch random remap (line 8:11).} After the local computation, each record is remapped to a new block ID chosen uniformly at random from $[1..M]$. As discussed earlier, this simplification removes branch-heavy load-balancing logic and keeps the hot path SIMD-friendly. The position map $\mathsf{MAP}$ is updated accordingly, and the record is appended to the stash queue $S[b']$ for its destination block $b'$, where it waits to be written back. 

\vspace{5pt}\noindent{\it Stage.4 Eviction (line 12:15).} The final phase drains the stash queues for all blocks in $\mathcal{B}$. For each block $b$ that was read in this round, we assemble the updated page image by merging the surviving contents of $\mathsf{page}_b$ with all records currently in $S[b]$. The stash for $b$ is then cleared, and the assembled page is written back in a single \texttt{WRITE} to \textcolor{blue}{$V_{\mathsf{host}}$}$[b]$. As a result, each block in $\mathcal{B}$ is touched exactly once per round, one read and one write, independent how many remapped records it absorbed from the stash. 

\vspace{5pt}\noindent{\bf Stash optimizations.} The stash is the temporary holding area for tuples that have been removed from a bucket but cannot yet be written back to their new locations. Because every ORAM access may touch the stash, it sits directly on the critical path. A naive GPU design would treat it as one dynamic container and linearly scan or physically reshuffle entries whenever a tuple is remapped or evicted. That is manageable on a CPU, but on a GPU it becomes expensive: remapping would move full tuple contents, eviction would touch unrelated entries, and updates would create exactly the kind of irregular memory traffic and synchronization pressure that GPUs handle poorly.

To address this, we use a decomposed storage that separate tuple payloads from stash metadata. The payload is the full node tuple, including node id, neighbor list, and embedding values, and it is stored in a contiguous array of fixed-size stash slots. The metadata only track where each tuple is stored and which bucket it is currently associated with. We maintain this metadata using two lightweight structures: a ring-buffer free pool and a reverse index. The free pool recycles slot identifiers, so insertion pops a free slot and eviction pushes it back. The reverse index is organized as per-bucket rows of \texttt{(node\_id, slot\_idx)} pairs. This design makes remapping cheap: a tuple usually stays in place, and only its reverse-index association moves from one bucket row to another. Eviction is also optimized. Instead of scanning the whole stash, we only scan the row for the current bucket, copy out the referenced slots, and return those slot ids to the free pool. In short, we move metadata eagerly and payloads only when they are actually consumed.

We also make stash access SIMD-friendly. A fully scalar stash path would be easy to implement, but it would serialize the most bandwidth-heavy part of ORAM access. Our design instead uses a split rule: one designated thread in the warp, lane~0, handles the small control-critical updates, such as popping or pushing slot ids in the free pool, clearing an old reverse-index entry, or writing a new \texttt{(node\_id, slot\_idx)} association after remap. The rest of the warp cooperatively performs row lookup, slot copy, dummy padding, and output materialization. To coordinate these steps, we use standard warp primitives such as \texttt{\_\_shfl\_sync}, which broadcasts a value from one lane to the others, and \texttt{\_\_ballot\_sync}, which collects per-lane predicates into a bit mask so the warp can quickly determine which entries matched or which slots are valid. This design keeps the logic correct without giving up throughput.

\section{\sys Analysis}
\label{sec:analysis}

We now present the formal analysis of \sys.

\vspace{4pt}\noindent{\bf Security analysis.} We first analyze the security guarantees of \sys. Recall that observable memory traces arise only from the host-resident final layer(s). Accordingly, the analysis reduces to the obliviousness of the CORAM host-access primitive.

\begin{theorem}\label{tm:security}
Let $K$ be the number of host blocks and let $M$ denote the maximum neighbor list size in HNSW. For any final-layer expansion over a neighbor list $\mathcal{V}$ with $|\mathcal{V}| \le M$, the CORAM host access in \sys satisfies the obliviousness definition in Definition~\ref{def:security}.
\end{theorem}
\begin{proof}[Proof.]
We prove obliviousness by showing that, for any two queries $q$ and $q'$ satisfying $\forall i \in \{0, \dots, k\},\ |\mathbf{r}_i| = |\mathbf{r}'_i|$, the distributions of their observable traces are computationally indistinguishable. We first consider the non-coalesced variant. Each logical node access is mapped to a host block chosen uniformly at random from $[K]$, and after every access the node is remapped independently to a fresh random block. Therefore, the sequence of block identifiers revealed during execution forms a sequence of independent uniform samples over $[K]$. Since $q$ and $q'$ induce traversal paths of the same length, their observable traces consist of the same number of such samples. In addition, if a requested node is served from the stash, the system issues a dummy block read, so the observable trace preserves both length and structure. As a result, the distributions of block access sequences under $q$ and $q'$ are identical.

\re{We now consider the coalesced variant. During the expansion of a node in the final HNSW layer, each neighbor appears at most once in the neighbor list. Therefore, each logical node is accessed at most once, and ORAM remapping introduces no dependencies among accesses within the same expansion.  Consequently, all accesses can be issued as a single batch without changing the logical execution. Since both $q$ and $q'$ produce sequences of equal length with identical distributions, the distribution over distinct block identifiers (or coalesced transcript) remains identical for both executions.}



Since all blocks are encrypted by TEE's memory encryption mechanism, and thus we say that any p.p.t. adversary cannot distinguish executions of $q$ and $q'$ unless the encryption itself fails.
\end{proof}

\vspace{4pt}\noindent{\bf Bandwidth complexity.} We next analyze host-I/O cost. We say that it is enough to study the bandwidth of one neighbor expansion invocation as the total bandwidth of a full query is linear in the number of such expansions.

\begin{theorem}\label{tm:band}
Let $N$ be the number of data points and let $K = \Theta(N)$ be the number of host pages. In the non-coalesced variant, the bandwidth of one neighbor expansion is at most
\[
O\!\left(M \cdot \frac{\log N}{\log \log N}\right)
\]
except with probability at most $1/N$.
\end{theorem}
\begin{proof}[Proof.]
As in prior ORAM analyses~\cite{guo2025bolt,path-oram,zheng2024h,ren2015constants}, the random-remapping process can be modeled as a balls-and-bins experiment: the $N$ nodes are balls, the $K$ host pages are bins, and each node is assigned independently and uniformly to one bin. Let $\ell_{\max}$ be the maximum page load. Standard balls-and-bins bounds~\cite{raab1998balls,mitzenmacher2017probability} show that when $K = \Theta(N)$, $\ell_{\max} = O\!\left(\frac{\log N}{\log \log N}\right)$
except with probability at most $1/N$. In one non-coalesced final-layer expansion, at most $M$ neighbor accesses are issued. Since the amount of data transferred by one logical access is bounded by the maximum block load, the total bandwidth is at most $M \cdot \ell_{\max}$ blocks.
\end{proof}
Note that the per-block bandwidth is upper bounded by $O\left(\frac{\log N}{\log \log N}\right)$, which provides a guideline for setting block sizes in practice. With this choice, the probability of overflow, i.e., that more data is randomly remapped to a block than it can accommodate, is at most $1/N$. For large $N$, this probability is negligible.

\begin{theorem}\label{tm:coalesced}
Let $K = \Theta(N)$ and assume $M = o(K)$. In one coalesced neighbor expansion, the expected number of host page reads is $M - \Theta(M^2/K)$. In particular, coalescing saves $\Theta(M^2/K)$ page reads in expectation compared to the $M$ reads in the non-coalesced variant.
\end{theorem}

\begin{proof}
We model the $M$ logical accesses as placing $M$ balls independently and uniformly into $K$ pages. Let $X$ be the number of distinct pages touched. A given page is accessed unless all $M$ accesses avoid it, which occurs with probability $(1 - 1/K)^M$. Thus, each page is touched with probability $1 - (1 - 1/K)^M$, and summing over all $K$ pages gives $\mathbb{E}[X] = K(1 - (1 - 1/K)^M)$. When $M = o(K)$, we use the expansion $(1 - 1/K)^M = 1 - M/K + \Theta(M^2/K^2)$, which yields $\mathbb{E}[X] = K(M/K - \Theta(M^2/K^2)) = M - \Theta(M^2/K)$. 
\end{proof}

\vspace{4pt}\noindent{\bf Stash analysis.} We now study the size of the stash.

\begin{theorem}[stash size]\label{thm:stash-sz}
Let $N$ be the number of tuples and $K$ the number of host blocks. Assume each logical access is drawn uniformly from the $N$ tuples, and each remapped tuple is reassigned independently and uniformly to one of the $K$ blocks. Let $\ell_{\max}$ be the max page size, and define $x^*:=\frac{NK}{N+K}$. Then, with probability at least $1-1/N$ the stash size is bounded by
\[
x^* + O\!\left(\ell_{\max}\sqrt{x^*\ln N}\right)
\]
\end{theorem}
\begin{proof}
We view the stash as a queue whose entries are labeled by their destination page in $[K]$. Let $X_t$ be the stash size after the $t$-th access. In one step, the stash changes for two reasons: one accessed tuple may be newly inserted into the stash, and some existing stash entries may be written back when their destination page is serviced. Accordingly, we write $X_{t+1}=X_t+A_t-D_t$, where $A_t\in\{0,1\}$ is the insertion indicator and $D_t$ is the number of tuples written back. Conditioned on $X_t=x$, exactly $x$ of the $N$ tuples already reside in the stash. Under the uniform-access assumption, the next logical access hits one of these tuples with probability $x/N$, so a fresh insertion happens with probability $1-x/N$. Thus $\mathbb{E}[A_t\mid X_t=x]=1-x/N$. At the same time, each of the $x$ queued tuples matches the currently processed page with probability $1/K$, so $\mathbb{E}[D_t\mid X_t=x]=x/K$. Combining the two terms gives the one-step drift
\(\mathbb{E}[X_{t+1}-X_t\mid X_t=x]=1-\frac{x}{N}-\frac{x}{K}\).

This expression already shows where the stash stabilizes: the equilibrium is the point where the expected change becomes zero. Solving $1-\frac{x}{N}-\frac{x}{K}=0$ gives $x^*=\frac{NK}{N+K}$. Moreover, if the stash is above this level, say $x=x^*+r$ for some $r>0$, then the drift becomes $\mathbb{E}[X_{t+1}-X_t\mid X_t=x^*+r]=-r\left(\frac{1}{N}+\frac{1}{K}\right)<0$. In other words, once the stash grows above $x^*$, the process has a restoring tendency that pulls it back.

We now extend this expectation bound into a tail bound. Since in one step, at most one tuple is inserted and at most $\ell_{\max}$ tuples are written back, so the stash size changes by at most $\ell_{\max}$ in absolute value. We then apply the same negative-drift concentration technique used in~\cite{guo2025bolt}, which yields $\Pr[X_t\ge x^*+u] \le \exp\!\left(-\Omega\!\left(\frac{u^2}{x^*\ell_{\max}^2}\right)\right)$. Setting $u=c\,\ell_{\max}\sqrt{x^*\ln N}$ for some constant $c>1$ gives $\Pr[X_t\ge x^*+u]\le 1/N$.
\end{proof}

This stash analysis shows that the stash size is bounded and does not grow indefinitely. We further validate this empirically in \S~\ref{sec:main-exp}.

\section{Evaluation}
\label{sec:evaluation}
We now discuss our evaluation of \sys to quantify its performance and scalability. We conduct end-to-end comparisons with prior work, analyze system bottlenecks, and examine cost breakdowns, memory usage, and overheads.

\subsection{Evaluation Setup}\label{sec:setup}
\vspace{5pt}\noindent{\bf \sys implementation.} We implement \sys in CUDA C++17, with about 8.1K lines of code (LoC), which keeps the trusted computing base (TCB) relatively small. We build the system using NVIDIA CUDA 12.9 (nvcc 12.9.86). By default, \sys offloads the final layer to the host, as this layer dominates the memory footprint. However, the design is not tied to this choice—since CORAM is generic, we can offload additional layers as well. We include separate experiments to evaluate different offloading configurations in \S~\ref{sec:off}. For host storage, we use a block size of 16 slots (i.e., $K = N/16$). Each block initially contains 8 real records and 8 dummy entries to maintain a uniform layout. We allocate the host storage region using \texttt{cudaMallocManaged} and configure it so that these pages are not cached in GPU HBM. This setup stresses the access path and lets us measure the worst-case performance when accessing layer-0 data. On the GPU, we use a \texttt{ChaCha20}-based random number generator for all randomness in CORAM.

\vspace{4pt}\noindent{\bf Baseline systems.}We use \textsc{Compass}~\cite{zhu2025compass} as our main baseline for end-to-end comparison. To our knowledge, it is the current state-of-the-art in oblivious semantic search with open-source artifacts. We also include \textsc{Bolt}~\cite{guo2025bolt} as a baseline ORAM design in our microbenchmarks, which focus specifically on ORAM performance. \textsc{Bolt} is, to our knowledge, the first design that leverages isolated HBM (\pmm) to accelerate ORAM operations. \re{We evaluate \textsc{Compass} and \sys on the same machine to enable a fair performance comparison. For \textsc{BOLT}, since we do not have access to the specialized hardware (e.g., the U55C FPGA) used in their evaluation, we report the best performance numbers from their original paper~\cite{guo2025bolt}. While obtained on different hardware, these results still provide a useful point of reference.}  

\vspace{4pt}\noindent{\bf Testbed.} All experiments run on a dual-socket server with two AMD EPYC 9124 CPUs (32 cores / 64 hardware threads in total) with SEV-SNP, 256GB of system memory, and an NVIDIA H100 PCIe GPU operating in CC mode. We deploy \sys inside a confidential VM (CVM), with the GPU configured in CC mode and passed through directly to the VM. 

\vspace{4pt}\noindent{\bf Datasets.} Unless otherwise specified, we use the same datasets as \textsc{Compass}, including \textsc{Laion100K}, \textsc{Sift1M}, \textsc{TripClick}, and \textsc{MSMarco}. Table~\ref{tab:datasets} summarizes their key characteristics. To construct our linear table representation, we implement a simple converter that takes graphs generated by \texttt{hnswlib}~\cite{hnswlib}, a widely-used HNSW library, and transforms them into our format. As shown in Table~\ref{tab:datasets}, our linear representation is smaller than the original \texttt{hnswlib} format, mainly because it eliminates pointer-based structures.

\begin{table}[tb]
\caption{Dataset Summary}
\label{tab:datasets}
\centering
\scalebox{0.8}{
\begin{tabular}{l r r r r r r}
\toprule
\textbf{Dataset} & \textbf{$N$} & \textbf{$M$} & \textbf{$d$} & \textbf{Size (\texttt{hnswlib})} & \textbf{Size (LTb)} & \textbf{\# Queries} \\
\midrule
LAION  & 100K  & 64  & 512 & 247\,MB & 210\,MB &  1,000\\
SIFT-1M     & 1M    & 64  & 128 & 996\,MB & 637\,MB &  10,000\\
TripClick   & 1.5M  & 128 & 768 & 5.9\,GB & 4.7\,GB &  1,175\\
MSMarco     & 8.8M  & 128 & 768 & 34\,GB  & 27\,GB  &  6,980 \\
\bottomrule
\end{tabular}
}
\end{table}

\subsection{End-to-End Performance Benchmark}\label{sec:main-exp}

\vspace{4pt}\noindent{\bf Comparison with \textsc{Compass}.} We first compare \sys with \textsc{Compass} in an end-to-end benchmark. Specifically, we run semantic queries on all four datasets and tune the recall (e.g., by adjusting $ef$ in HNSW) to reach a target threshold. We then measure the average query time. Figure~\ref{fig:endtoend} shows the results.
\begin{figure}[tb]
    \centering
\includegraphics[width=0.9\linewidth]{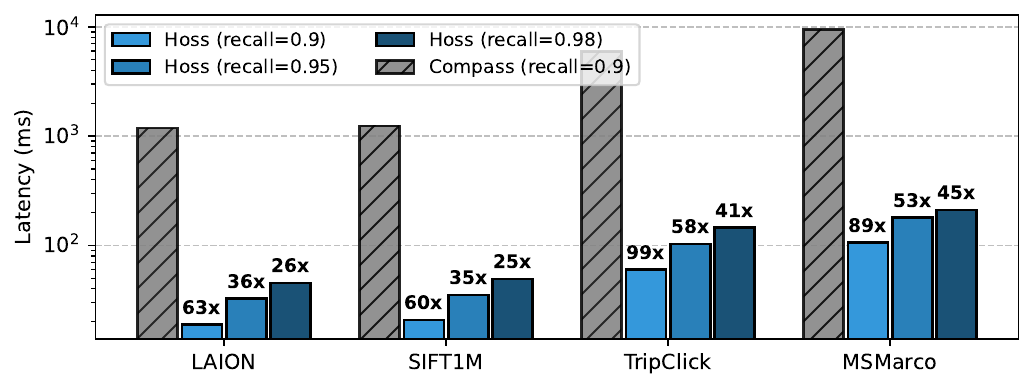}
    \caption{\re{End-to-end query performance comparison.}}
    \label{fig:endtoend}
    \vspace{-1em}
\end{figure}

From Figure~\ref{fig:endtoend}, we see that \sys outperforms \textsc{Compass} across all settings, \re{achieving at least $60\times$ speedup at 0.9 recall and up to $99\times$.} We also observe that the performance gap widens at larger scales (e.g., \textsc{TripClick} and \textsc{MSMarco}). \textsc{Compass} only reports results up to 0.9 recall, while \sys can operate at higher accuracy levels. We therefore also evaluate \sys at 0.95 and 0.98 recall by increasing the HNSW parameter $ef$. Even at 0.98 recall, \sys still achieves significant speedups over \textsc{Compass} at 0.9 recall, for example, \re{up to $45\times$ on \textsc{MSMarco}.} This shows significant improvements of \sys over the SOTA oblivious semantic search system. 

Since \sys and \textsc{Compass} run on different hardware, we also report normalized slowdowns to isolate the cost of oblivious primitives (Figure~\ref{fig:breakdown}). Each system is normalized to its own non-private baseline, which factors out hardware effects and focuses on oblivious overheads. The \sys baseline uses a GPU kernel that preserves the same HNSW logic and offloading as \sys, but accesses data directly without CORAM primitives. For \textsc{Compass}, we implement the same HNSW algorithm on CPU and use it as the baseline.
\begin{figure}[tb]
    \centering
\includegraphics[width=0.9\linewidth]{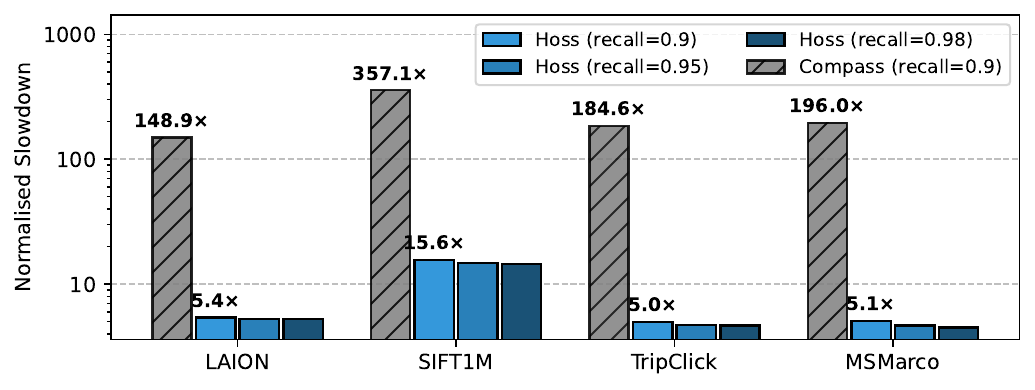}
    \vspace{-.5em}
    \caption{\re{Normalized slowdown comparison.}}
    \label{fig:breakdown}
    \vspace{-.5em}
\end{figure}

We observe that \textsc{Compass} exhibits substantially larger normalized slowdowns across all groups, often exceeding 100$\times$ and \re{up to 357$\times$}, whereas \sys remains below 6$\times$ in most cases. For both systems, SIFT1M shows the largest slowdown. This is due to its lower dimensionality (128 vs.\ 768), which reduces data movement cost and makes the ORAM overhead, such as remapping and position lookup, more prominent. In higher-dimensional datasets, data movement dominates execution time in both the private system and its baseline, so the additional cost of ORAM contributes a smaller fraction, leading to lower normalized slowdowns. This also explains the trend in \sys at higher recall levels. As recall increases, both \sys and its non-private baseline become dominated by data movement and copying. The relative impact of ORAM overhead therefore decreases, resulting in slightly lower normalized slowdowns.


\vspace{4pt}\noindent{\bf Memory usage experiments.} We next examine the memory usage of \sys. Most components are fixed and can be reported directly. The stash, however, is dynamic: its size depends on accesses but is theoretically bounded with high probability. To validate this in practice, we run a stress experiment that performs a large number of CORAM accesses (10$\times$ of data sizes) and tracks the maximum stash size over time (Figure~\ref{fig:stash}).  We then report the max stash usage together with static size components (e.g., position maps, host storage, etc.) to provide a complete memory breakdown (Table~\ref{tab:memory}).
\begin{figure}[tb]
    \centering
\includegraphics[width=0.9\linewidth]{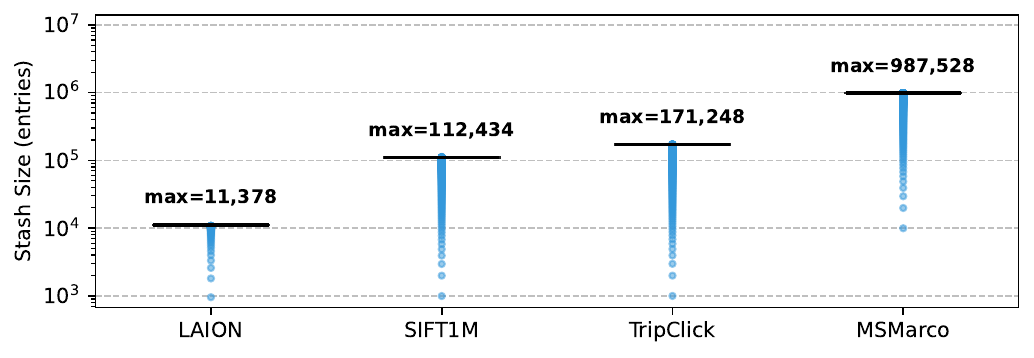}
    \caption{Stash sizes over 10$\times$ data\_size accesses.}
    \label{fig:stash}
    \vspace{-1em}
\end{figure}

\begin{table}[tb]
\caption{Complete memory breakdown}
\label{tab:memory}
\centering
\scalebox{0.85}{
\begin{tabular}{l r r r r r}
\toprule
\textbf{Dataset} & \textbf{Pos.\ Map} & \textbf{Buckets} & \textbf{Max Stash} & \textbf{Total} & \textbf{GPU \%} \\
\midrule
LAION  & 0.4\,MB  & 440.2\,MB  & 25.0\,MB  & 465.6\,MB  & 5.5\% \\
SIFT-1M     & 3.8\,MB  & 1.44\,GB   & 82.8\,MB  & 1.52\,GB   & 5.6\% \\
TripClick   & 5.8\,MB  & 10.18\,GB  & 586.0\,MB & 10.76\,GB  & 5.4\% \\
MSMarco     & 33.7\,MB & 59.09\,GB  & 3.30\,GB  & 62.42\,GB  & 5.3\% \\
\bottomrule
\end{tabular}
}
\end{table}

Figure~\ref{fig:stash} shows that stash sizes remain bounded even under sustained stress accesses. Throughout the entire run, the maximum stash size stays well below $12\%$ of the total data entries (raw data). From the full memory breakdown (Table~\ref{tab:memory}), we observe that GPU memory (position map + maximum stash) accounts for only about $5.5\%$ of total storage. This is because host blocks are $2\times$ over-provisioned, with half of each block occupied by dummy entries. Overall, \sys uses GPU memory efficiently, requiring only a small \pmm footprint to support large datasets.

Our machines have limited host memory, which prevents us from conducting stress tests to evaluate the extreme scalability that \sys can support. However, based on the observed memory characteristics, we can project its scalability and provide guidance for modern data center environments, where servers may be equipped with very large host memory (e.g., exceeding 4TB~\cite{itpro_dl340_gen12_2026}). We assume the same bucket size (e.g., 16 entries), so the GPU memory fraction remains stable at around $5.5\%$ (as by Theorem~\ref{thm:stash-sz}). We then estimate the raw data size as half of the host bucket region to quantify the max data size supported under full GPU \pmm utilization. The projection figures are in Table~\ref{tab:projection}. We can see that, with full GPU \pmm utilization and higher-end GPUs such as GH200-NVL, \sys has potential to scale to terabytes semantic data.
\begin{table}[tb]
\caption{Projected capacity under full GPU HBM usage. We select GPUs that support GPU TEE mode.}
\label{tab:projection}
\centering
\scalebox{0.9}{
\begin{tabular}{l r r r}
\toprule
\textbf{Platform} & \textbf{GPU HBM} & \textbf{Host} & \textbf{Raw Data} \\
\midrule
A100~\cite{nvidia_a100}      & 40\,GB  & 694\,GB   & 347\,GB \\
H100~\cite{nvidia_h100}      & 80\,GB  & 1.39\,TB  & 694\,GB \\
H200~\cite{nvidia_h200}      & 141\,GB & 2.45\,TB  & 1.22\,TB \\
GH200 NVL2~\cite{nvidia_gh200} & 288\,GB & 5.00\,TB  & 2.50\,TB \\
\bottomrule
\end{tabular}
}
\end{table}

\subsection{CORAM Micro-benchmark}\label{sec:microbench}
In this section, we conduct microbenchmarks to evaluate the performance of our CORAM design and compare it with the SOTA ORAM system, \textsc{Bolt}, which also leverages \pmm. 

\vspace{4pt}\noindent{\bf Comparison with \textsc{Bolt}.} We use the same benchmark setup as \textsc{Bolt}, which includes two sets of experiments. The first varies the number of data entries while keeping the tuple size fixed, and measures the average ORAM access latency. The second fixes the dataset size and varies the value size. For consistency, we use the same data format and scale settings as \textsc{Bolt}, where each ORAM entry is an address–value pair with a 4B address and an 8B value by default. Figure~\ref{fig:bolt} reports the results of both experiments.
\begin{figure}[tb]
    \centering
\includegraphics[width=.99\linewidth]{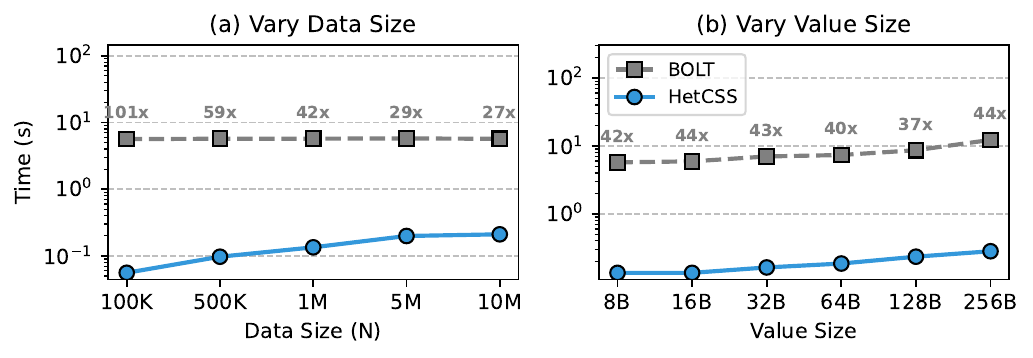}
    \caption{Performance comparison with BOLT.}
    \label{fig:bolt}
    \vspace{-1em}
\end{figure}

Figure~\ref{fig:bolt}(a) shows the results when varying the number of data entries. Across all settings, \sys consistently outperforms \textsc{Bolt}, achieving up to $101\times$ speedup. This gain comes primarily from our linear table layout, which allows direct position map lookups, while \textsc{Bolt} relies on a hash table mechanism~\cite{guo2025bolt}. 

We observe that the performance gap narrows as the dataset size increases. This is expected, as \textsc{Bolt} achieves an asymptotic bandwidth cost of $O(\log\log N)$, while ours is $O(\tfrac{\log N}{\log\log N})$. Nevertheless, \textsc{Bolt}'s design does not map well to GPU execution, and \sys still maintains a substantial advantage, with speedups of up to $27\times$.

Figure~\ref{fig:bolt}(b) shows the value scaling experiment. Here, we can see that across all settings, \sys delivers consistent improvements over \textsc{Bolt}, with speedups of up to $44\times$.

\vspace{4pt}\noindent{\bf Performance gains from coalescing.} A key design contribution of CORAM is its coalescing mechanism. \re{To isolate its impact, we implement two ablation variants: one disables coalescing while preserving batched execution, and the other further disables batching, issuing all accesses sequentially. We compare these two variants with the full CORAM design in terms of access latency.} We use the same two scaling experiments as before, but with data that better reflects semantic workloads. Unless otherwise specified, we use a dataset with 1M entries and 128-dimensional embeddings. In the value-scaling experiment, we vary the embedding dimension from 32 to 768. Figure~\ref{fig:batch} shows the results.
\begin{figure}[tb]
    \centering
\includegraphics[width=.99\linewidth]{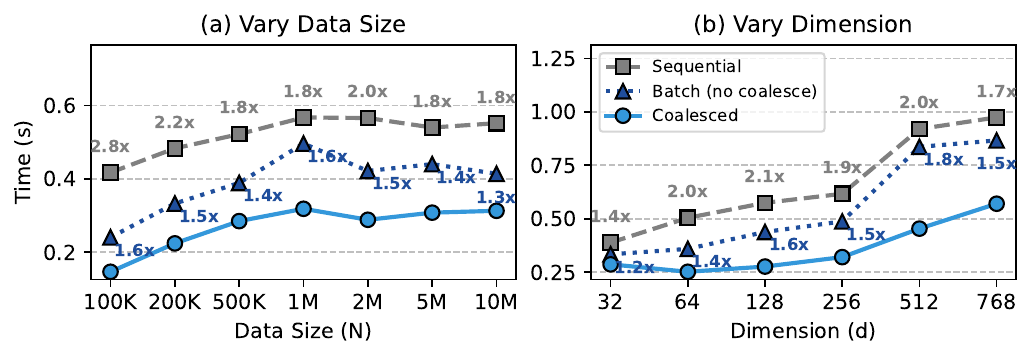}
    \caption{\re{Performance comparison of different ORAM modes.}}
    \label{fig:batch}
    \vspace{-1em}
\end{figure}

From both figures, we see that CORAM consistently outperforms the both variants. \re{The speedup reaches up to $2.8\times$ over fully sequential, and $1.8\times$ over batched mode. Comparing CORAM with the batched variant without coalescing, we observe that the benefit of coalescing becomes more pronounced as the embedding dimension increases. This is expected because larger embeddings incur higher I/O costs, allowing coalescing to eliminate more redundant memory accesses and thus achieve greater performance gains.}

\subsection{Offloading Mode Experiments}\label{sec:off}
In this section, we evaluate \sys under different offloading strategies. In addition to offloading only the final layer, we progressively offload more layers to the host, up to the extreme case where only metadata (e.g., position map and stash) reside in the GPU TEE's \pmm. This experiment requires reconfiguring the system and storage layouts across settings. To keep it manageable, we focus on the \textsc{Sift1M} workload, which is sufficient to capture the trends under different offloading strategies. Figure~\ref{fig:offload} shows the performance under these configurations.
\begin{figure}[tb]
    \centering
\includegraphics[width=0.9\linewidth]{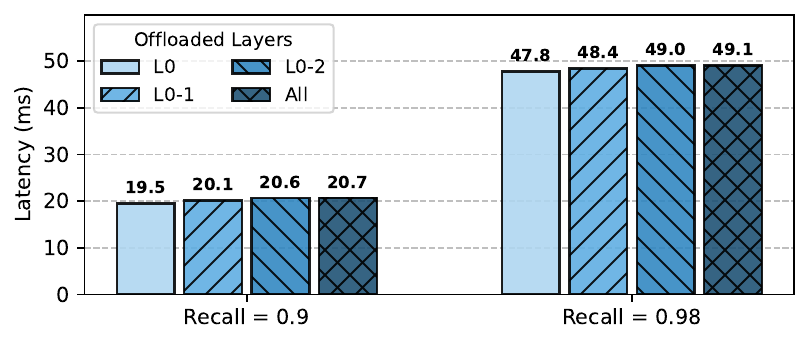}
    \caption{Offloading experiments.}
    \label{fig:offload}
\end{figure}

From the figure, offloading additional layers introduces extra overhead and increases query latency. However, the increase is modest compared to only L0 offloading. For example, in the recall = 0.9 setting, offloading all layers adds only 6\% overhead relative to offloading L0 alone.
This behavior is expected for two reasons. First, HNSW layers shrink exponentially toward the top, so while offloading more layers increases ORAM accesses, the incremental cost per layer is exponentially smaller. Second, HNSW performs far fewer node accesses in upper layers: these layers use simple greedy traversal with small candidate sets, whereas the final layer relies on beam search and explores many more nodes.
This also explains why at recall = 0.98 the additional overhead is even smaller: achieving higher recall makes the final-layer search more expensive, amortizing the cost of offloading upper layers. Overall, offloading more layers leads to only minor performance degradation.
\section{Discussion}\label{sec:discussion}
We briefly discuss extensions enabled by \sys to illustrate its broader applicability. A full exploration is beyond the scope of this paper and left to future work.

\vspace{4pt}\noindent{\bf Volume and timing hiding.} Standard obliviousness definitions~\cite{path-oram,goldreich1987towards,ren2015constants,resizable-tree-based-oram} focus on making access sequences indistinguishable when they have the same length. This means they do not capture timing or volume leakage~\cite{kellaris2016generic,cash2015leakage,blackstone2019revisiting,chuang2026tee,grubbs2018pump,oya2021hiding}, where different queries naturally lead to different amounts of work. In practice, these leakages are often handled with separate, orthogonal techniques~\cite{amjad2021dynamic,wang2021dp,wang2022incshrink,bater2018shrinkwrap,patel2019mitigating}. \sys follows the standard definition of obliviousness, but can also be extended to address timing and volume leakage. Intuitively, for the upper-layer searches running on the GPU, the access patterns are hidden, so volume leakage is not a concern, but timing can still vary. A possible fix is to pad execution to a fixed upper bound so that all queries take the same time before moving to the final layer. For offloaded layers, volume leakage becomes more visible. For instance, in HNSW beam search~\cite{malkov2018efficient}, nodes visited are skipped, which creates data-dependent volume patterns. One way to handle this is to always fetch candidate nodes and do the filtering inside the GPU, so externally the behavior looks the same. Similarly, we can remove early stopping and instead run a fixed number of exploration steps to smooth out timing differences.

\vspace{4pt}\noindent{\bf Inter-query parallelism.} Inter-query parallelism~\cite{ganguly1992query} is a appealing feature in many retrieval systems and modern database engines. Unlike intra-query parallelism, which focuses on accelerating a single query as in \sys, this approach executes multiple, potentially heterogeneous queries concurrently to improve overall throughput. This design is particularly attracting for GPU-based systems, where a single query often cannot fully utilize the available parallel resources~\cite{douze2024faiss, chen2021spann}. However, supporting inter-query parallelism is fundamentally challenging in ORAM contexts, as na\"ive parallelism can easily violate security guarantees~\cite{chakraborti2018concuroram,asharov2022optimal,lorch2012toward}. The coordination logic usually needs carefully design and made oblivious as well~\cite{chakraborti2018concuroram}.

That said, the GPU’s \pmm opens up new opportunities. As shown in our CORAM design, parallel accesses can be efficiently coordinated within \pmm. While our current setting does not consider identical accesses within a batch (e.g., the same node), the design can be extended to support general parallel processing. In particular, when such conflicts arise, we can internally replay the ORAM logic: issuing a single real access with remapping, while serving the remaining requests as randomized dummy accesses. This preserves ORAM invariants while enabling batched execution, providing a foundation for efficient inter-query parallelism.

\vspace{4pt}\noindent{\bf Adapting to other oblivious systems.} Although \sys is designed for HNSW-based semantic search, its core ideas are not tied to a specific algorithm or workload. In particular, the CPU-bypassed execution model and GPU-assisted ORAM are general abstractions that apply beyond HNSW and even beyond semantic search. For example, \sys can support other graph-based algorithms~\cite{ootomo2024cagra,shi2018graph}, as long as the graph can be expressed in our linear table representation. Adapting to a new graph workload primarily requires changing the access interface while reusing the same ORAM backend. More broadly, the ORAM abstraction itself is not specialized to graphs, but rather a batched key-value access interface, making it straightforward to map a wide range of algorithms onto it. Beyond graph workloads, these abstractions naturally extend to other oblivious systems, such as relational databases~\cite{eskandarian19obliDB,zheng2017opaque,special2024} and time-series databases~\cite{dauterman2022waldo,faisal2023tva}. Overall, \sys provides a general and reusable foundation for building high-performance oblivious systems across diverse application domains.

\re{\vspace{4pt}\noindent{\bf Dynamic index support.}
Although this work focuses on secure search over a pre-built HNSW index, CORAM is not limited to static indexes. Dynamic index-management primitives naturally build on the same search primitive. We assume the ORAM storage is pre-allocated with sufficient free space for future growth. For example, insertion first performs an oblivious search to identify the insertion neighbors, followed by a bounded number of graph updates~\cite{malkov2018efficient}. Deletion can leverage the standard lazy-deletion mechanism~\cite{hnswlib} by obliviously marking a node as deleted, and the reclaimed space can be reused by subsequent insertions. If the reserved ORAM capacity is eventually exhausted due to continuous growth, the data owner can periodically rebuild the index.}

\section{Related Work}
\label{sec:related}
\noindent{\bf ORAM and oblivious data systems.}
ORAMs has long been the canonical abstraction for hiding memory-access patterns on untrusted storage~\cite{goldreich1987towards,path-oram,ren2015constants,asharov2020optorama,asharov2023futorama, goldreich1996software, zheng2024h, tinoco2023enigmap, asharov2022optimal, pinkas2010oblivious, bindschaedler2015practicing, wang2014oblivious}. A parallel line of work developed the oblivious algorithms with which execution's memory accesses do not depend on inputs, which covers sorting networks~\cite{batcher1968sorting,ajtai19830,goodrich2014zig, asharov2020bucket}, shuffle~\cite{ghazi2019scalable,sasy2022fast} and more advanced data processing~\cite{dittmer2020oblivious, special2024, wang2021dp, wang2022incshrink, bater2018shrinkwrap, chu2021differentially, qin2022adore, chang2022towards}. As obliviousness moved from cryptographic theory into systems, in recent years, there has made a series end-to-end data processing systems that delivers strigent obliviousness guarantees, such as \textsc{ZeroTrace}~\cite{sasy18zerotrace}, \textsc{Oblix}~\cite{oblix}, \textsc{ObliDB}~\cite{eskandarian19obliDB}, and many more~\cite{chamani2023graphos, zheng2024h,guo2025bolt,Ahmed2025oasisdb, qin2022adore, qiudoquet}. \textsc{Compass}~\cite{zhu2025compass} is the SOTA oblivious semantic search system that to achieve high-accuracy encrypted semantic search~\cite{zhu2025compass}. Nevertheless, these works rely on the traditional assumption of only an $O(1)$-sized \pmm, or require a trusted proxy or client~\cite{asharov2022optimal, zhu2025compass}. As a result, they remain subject to classical ORAM lower bounds and the associated performance overheads.

\vspace{4pt}\noindent{\bf Secure memory hardware.} Due to the high cost of ORAM and other oblivious primitives, prior work has explored secure memory hardware~\cite{aga2017invisimem, awad2017obfusmem, oh2020trustore, duy2022se, choi2024shieldcxl, guo2025bolt, xu2019hermetic} that conceals memory channels and directly hides access patterns without relying on oblivious primitives. However, these designs typically target limited-capacity memory resources, such as registers~\cite{xu2019hermetic}, BRAMs~\cite{oh2020trustore, awad2017obfusmem, aga2017invisimem}, or specialized memory nodes~\cite{duy2022se, choi2024shieldcxl}. As a result, they are either not general-purpose or provide only limited capacity. More recent work moves to on-package HBM to hide access patterns, leading to a new class of systems~\cite{volos2018graviton, hunt2020telekine, guo2025bolt} that are significantly faster than traditional oblivious designs. \sys follows this direction by leveraging HBM as \pmm, while addressing unique challenges in supporting semantic searches through significant co-design. To our knowledge, this is the first work of its kind.


\vspace{4pt}\noindent{\bf Accelerator TEEs.} TEE designs have traditionally focused on CPUs, with prominent examples including Intel SGX~\cite{costan2016intel}, AMD SEV~\cite{amd_sev}, and Arm TrustZone~\cite{pinto2019demystifying}. With the rise of accelerators as a central component in modern data centers, recent work has expanded TEE support to these platforms. In particular, GPU-based TEEs have emerged as the dominant direction~\cite{zhao2025gpu, mai2023honeycomb, gu2026blueprint, deng2022strongbox, vaswani2022confidential, wang2026sok, wu2023building, chrapek2026secperf}, with production deployments such as NVIDIA’s confidential computing offerings~\cite{nvidia2025secureai}. FPGA-based TEEs~\cite{guo2025bolt, zhao2022shef, armanuzzaman2022byotee} target specialized low-latency confidential workloads, while ASIC-based TEEs~\cite{vaswani2022confidential, vaswani2023confidential, dhar2024ascend} integrate security directly into neural processing units for protecting ML pipelines. These works primarily ensure the confidentiality of data values. However, they do not directly address leakage through data-dependent execution, where memory accesses or control flow reveal sensitive information. \sys builds on accelerator-based TEEs and complements them with mechanisms that rigorously control data-dependent behavior, thereby reducing leakage from access patterns.



\section{Conclusion}
\label{sec:conclusion}

Oblivious semantic search faces a fundamental tension between strong privacy and practical performance. We revisit this tension in the context of GPU TEEs and show that large on-package HBM opens a new path. By treating HBM as a larger \pmm, we reduce reliance on expensive ORAM operations. Guided by this insight, we build \sys, a heterogeneous TEE-based system that combines GPU and CPU enclaves for efficient, scalable, and secure HNSW-based retrieval. The design confines data-dependent execution to GPU-resident \pmm and uses a redesigned ORAM interface for host memory to protect data contents and access patterns without prohibitive overhead. \sys shows that strong security can coexist with high performance and scalability.
\bibliographystyle{ACM-Reference-Format}
\bibliography{ccs-sample}

\appendix 

\section{Open Science} 


\label{app:open-science}

The artifacts necessary for evaluating the contributions of \sys are available in an anonymous repository:
\url{https://anonymous.4open.science/r/hetcss-7758/}.
The repository includes the implementation, experiment scripts, benchmark workloads, configuration files, dataset preparation scripts, plotting scripts, and instructions for reproducing the main experimental results reported in the paper.

\section{Ethical Considerations}

This work proposes Hoss, a privacy-enhancing system for oblivious semantic search using heterogeneous CPU--GPU TEEs. The goal of the system is to protect outsourced embedding datasets and query-dependent access patterns from an untrusted cloud software stack. Therefore, the intended impact of the work is positive, particularly for sensitive search workloads such as biomedical retrieval, enterprise search, and private retrieval-augmented generation.

This work does not involve human subjects, user studies, or the collection of new personal data. Our evaluation uses benchmark datasets and system-level performance measurements, and we do not collect or analyze private user queries. The work also does not attack deployed third-party systems or disclose new vulnerabilities. The adversarial analysis is limited to a clearly specified threat model and is used only to motivate and evaluate a defensive system.

The main ethical risk is possible over-interpretation of the security guarantees. To mitigate this risk, the paper explicitly states its threat model, security goals, and infrastructure assumptions. Hoss assumes standard TEE protections, encrypted inter-TEE communication, and dedicated GPU TEE execution, and it does not claim protection against out-of-scope threats such as availability attacks, malicious inputs, or invasive physical side channels.

Like other privacy-enhancing technologies, Hoss could in principle be used in undesirable applications. However, the techniques presented in this paper are intended to reduce privacy leakage in legitimate outsourced search deployments. We believe the benefits of enabling efficient confidential semantic search outweigh the limited dual-use risks, provided that real-world deployments follow applicable legal, organizational, and access-control requirements.

\section{Generative AI Usage}
The authors used ChatGPT for minor editorial assistance, including grammar checking, spelling correction, and light style polishing. No substantive scientific content, citations, or technical claims were generated by the tool. All AI-assisted edits were manually reviewed and verified by the authors.
\end{document}